\documentclass[11pt]{article}

\RequirePackage[OT1]{fontenc}
\usepackage{amsthm,amsmath,natbib}
\usepackage{multirow}
\RequirePackage[colorlinks,citecolor=blue,urlcolor=blue]{hyperref}

\def\frac#1#2{{\textstyle{#1\over#2}}}

\DeclareSymbolFont{AMSb}{U}{msb}{m}{n}
\DeclareMathSymbol{\Natural}{\mathbin}{AMSb}{"4E}
\DeclareMathSymbol{\Integer}{\mathbin}{AMSb}{"5A}
\DeclareMathSymbol{\Real}{\mathbin}{AMSb}{"52}
\DeclareMathSymbol{\Rational}{\mathbin}{AMSb}{"51}
\DeclareMathSymbol{\Imaginary}{\mathbin}{AMSb}{"49}
\DeclareMathSymbol{\Complex}{\mathbin}{AMSb}{"43} 
\DeclareMathSymbol{\Disk}{\mathbin}{AMSb}{"44} 
\def\bi{\begin{itemize}}
\def\ei{\end{itemize}}
\def\bd{\begin{description}}
\def\ed{\end{description}}
\def\ben{\begin{enumerate}}
\def\een{\end{enumerate}}
\def\bv{\begin{verbatim}}
\def\ev{\end{verbatim}}

\def\hat#1{{\widehat{#1}}}

\def\2to{{\ {\buildrel 2\over \longrightarrow}\ }}

\def\I1ton{{$I_1,\ldots,I_n$}}
\def\X1ton{{$X_1,\ldots,X_n$}}
\def\Y1ton{{$Y_1,\ldots,Y_n$}}
\def\Z1ton{{$Z_1,\ldots,Z_n$}}
\def\R1ton{{$R_1,\ldots,R_n$}}
\def\e1ton{{$e_1,\ldots,e_n$}}
\def\t1ton{{$t_1,\ldots,t_n$}}
\def\x1ton{{$x_1,\ldots,x_n$}}
\def\y1ton{{$y_1,\ldots,y_n$}}
\def\z1ton{{$z_1,\ldots,z_n$}}

\newtheorem{defn}{Definition}

\newtheorem{lemma}[defn]{Lemma}
\newtheorem{example}[defn]{Example}

\usepackage[english]{babel}
\usepackage{amsfonts}
\usepackage{amssymb}
\usepackage{graphicx}
\usepackage{enumerate}
\usepackage{bm}

\def\bi{\begin{itemize}}
\def\ei{\end{itemize}}
\def\be{\begin{equation}}
\def\ee{\end{equation}}

\begin{document}
\thispagestyle{empty}
\baselineskip=28pt
\vskip 5mm
\begin{center} {\Large{\bf Spatial random field models based on L\'evy indicator convolutions}}
\end{center}

\baselineskip=12pt

\vskip 5mm

\newtheorem{definition}{Definition}
\newtheorem{proposition}{Proposition}

\begin{center}
\large
Thomas Opitz$^1$
\end{center}

\footnotetext[1]{
\baselineskip=10pt  BioSP, INRA, 84914, Avignon, France. E-mail: Thomas.Opitz@inra.fr}

\baselineskip=17pt
\vskip 4mm
\centerline{\today}
\vskip 6mm

\setlength{\bibsep}{1pt}
\begin{center}
{\large{\bf \noindent Abstract }}
\end{center}
Process convolutions  yield random fields with flexible marginal distributions and dependence beyond Gaussianity, but statistical inference is often  hampered by a lack of closed-form marginal distributions, and  simulation-based inference may be prohibitively computer-intensive. We here remedy such issues through  a class of process convolutions  based on smoothing a ($d+1$)-dimensional L\'evy basis  with an indicator function kernel  to construct a $d$-dimensional convolution process.  Indicator kernels ensure univariate distributions in the  L\'evy basis family, which provides a sound basis for  interpretation, parametric modeling  and statistical estimation. We propose a class of isotropic stationary convolution processes constructed through hypograph indicator sets defined  as the space  between the curve $(s,H(s))$ of a spherical probability density function $H$ and the plane $(s,0)$.  If $H$ is radially decreasing, the covariance is expressed  through the univariate distribution function of $H$. The bivariate joint tail behavior in such convolution processes  is studied in detail. Simulation and modeling extensions beyond isotropic stationary spatial models are discussed, including latent process models. For statistical inference of parametric models, we  develop  pairwise likelihood techniques and illustrate these on  spatially indexed weed counts in the Bjertop data set,  and on daily wind speed maxima observed over $30$ stations in the Netherlands. 
\par\vfill\noindent


\pagenumbering{arabic}


\newpage



\section{Introduction}\label{sec:introduction}
Gaussian processes are widely used in spatial statistics thanks
to their nice theoretical and statistical properties, but they 
may present critical shortcomings when features such as non-Gaussian marginal
distributions,  asymmetry in univariate and multivariate lower and upper tails  or relatively strong  dependence at extremal levels  are crucial when modeling data. 
Some of these shortcomings may be remedied by applying marginal transformations  or by embedding a random variable for  the mean or the scale of the process  \citep[e.g.,][]{Roislien.Omre.2006,Krupskii.al.2015,Huser.al.2017},  but then statistical inference may become more challenging, and the model may not be identfiable if spatial data are not replicated in time. 
Hierarchical models including unobserved  Gaussian components have become popular with Bayesian inference approximating the posterior distribution \citep[e.g.,][]{Rue.al.2009,Banerjee.al.2014}, but  we often  lack closed-form distribution functions and densities.  Moreover, such models often rely on conditional independence assumptions with respect to latent components, leading to dependence strength that is even weaker than the Gaussian one. Spatial copula approaches \citep{Joe.2014}  propagate a full separation  of marginal distributions
and of dependence in spatial data. This leads to considerable flexibility,  but then interpreting covariance structure becomes less intuitive, inference may be tricky when some parameters influence  both margins and dependence, and one may view this as abandoning to seek modeling of the true data generating mechanism \citep{Mikosch.2006}. 

To remedy such modeling issues,  process convolution methods have been proposed as a powerful alternative  to generate models for stochastic processes with flexible and complex dependence
structures \citep[\emph{e.g.},][]{Higdon.1998,Brix.1999,Higdon.2002,Hellmund.al.2008,Calder.2008}.
However, in practice, most approaches focus on
Gaussian processes, whose convenient 
stability properties preserves  multivariate Gaussian distributions in the
convolution process. Non-Gaussian source processes have been used
 \citep[\emph{e.g.},][]{Kozubowski.al.2013,Jonsdottir.al.2013,BarndorffNielsen.al.2014,Bousset.al.2015,Noven.al.2015}, but they are often restricted to
simulation-based inference to remedy the absence of
closed-form expressions for  marginal distributions, whose interpretation
also becomes less intuitive.
In contrast, our aim here is  to  develop a general non-Bayesian framework of 
spatial process convolution models with flexible and easily interpretable marginal
distributions and dependence structures, amenable to inference based on 
covariograms or composite likelihoods. 


The general idea of kernel convolutions is to smooth a random  measure $L(\cdot)$ through a nonnegative kernel  $k$ to generate the convolution process $X(s)=\int k(s,y)L(\mathrm{d}y)$. The source measure $L$ assigns measure $L(A)=\int_A L(\mathrm{d}s)$ to bounded Borel sets $A$.  For the source process $L$ to behave like white noise over  continuous space,  it must be a \emph{L\'evy process} characterized by independent increments \citep{Sato.1999,BarndorffNielsen.al.2012}, and the collection of random variables $L(A)$ for abritrary sets $A$ is called a \emph{L\'evy basis}. 
The random variables $L(A)$  have infinitely divisible distribution, for instance  the stable distributions (Gaussian, Cauchy, ...) or the gamma, inverse Gaussian, compound Poisson, Poisson  or negative binomial distributions. Conversely, to construct a white noise source process over continuous space,   we may fix  an infinitely divisible distribution $F_A$  for the random measure $L(A)$ of  a fixed bounded Borel set $A$, and thanks to the infinitely divisible structure we can then define a stationary and isotropic white noise process $L(\cdot)$. The property of infinite divisibility  allows us to abstract away  from the any grid discretization of space. 
Our key proposal here is to use stationary kernels defined as indicator functions $k(s,y)=1_A(s-y)$ for some fixed indicator set $A$, such that the margins of the convolution process  have distribution in the Levy basis:  $X(s)\sim F_A$. Therefore, tractable L\'evy basis distributions $F_A$ are preserved in the convolution process, which is not be possible when using more general kernels on L\'evy bases such as the gamma or inverse Gaussian ones, or on any other L\'evy basis that is not stable. Moreover, modeling count data requires an integer-valued support; using discrete L\'evy basis such as the Poisson or negative binomial ones in our indicator kernel approach allows generating integer-valued spatially dependent processes. 

This flexible spatial modeling framework using indicator kernels may be viewed as an extension to  \emph{trawl processes}, which have been proposed for time series modeling
 \citet{BarndorffNielsen.al.2014,Noven.al.2015,Grahovac.al.2017}. To go beyond the
very limited range of dependence structures in $\mathbb{R}^d$ when smoothing a L\'evy basis defined in $\mathbb{R}^d$ with an indicator kernel, we will instead use a L\'evy basis in 
 $\mathbb{R}^{d+1}$ with an indicator set defined as a \emph{hypograph}, i.e.,  as the  area between the hyperplane  $s_{d+1}=0$ and the surface $(s,H(s))$ with $s=(s_1,\ldots,s_d)^T$ for
 a nonnegative function $H$ called  \emph{height function}. Radially symmetric height functions $H$ correspond to densities of a
 spherical probability distribution, which provides easily interpretable  dependence models. If further $H$  is non-increasing,we get nice and simple relationships in terms of the univariate distribution function of $H$ between the set covariance function of the hypograph, the covariance function of $X(s)$ and the tail correlation function. 

Section \ref{sec:levy} presents generalities of the  L\'evy convolution approach. Modeling with   hypographs is developed in Section \ref{sec:hypograph}; simulation techniques are shortly discussed in Section \ref{sec:simulation}.  Section \ref{sec:tail}   is dedicated to  theoretical results on the tail behavior of the convolution process.    Extensions to anisotropy  and space-time modeling are treated in Section \ref{sec:specificmodels}, which also puts forward several flexible latent process models. We develop parametric inference  with a focus on pairwise likelihood in Section \ref{sec:inference}. The two applications of Section \ref{sec:application} demonstrate the versatility of our modeling framework, one to weed counts (Bjertop farm data),  the other to Netherlands daily wind speed measurements. Concluding remarks and some perspectives are summarized in Section \ref{sec:discussion}. Proofs are deferred to the Appendix. 
\section{L\'evy convolution processes}
\label{sec:levy}
\subsection{L\'evy bases}
\label{sec:levybasis}
A L\'evy basis $L$ is  a collection of random variables indexed by sets. 
  We denote the Borel-$\sigma$-algebra  in $\mathbb{R}^d$  by $\mathcal{B}(\mathbb{R}^d)$, and $\mathcal{B}_b(\mathbb{R}^d)$ refers to  its  sets with finite Lebesgue measure. We just write $\mathcal{B}$ and $\mathcal{B}_b$ respectively when the reference to $\mathbb{R}^d$ is clear.  We call \emph{random noise} any collection of $\mathbb{R}$-valued random variables $\{L(A)\mid A \in \mathcal{B}_b\}$  possessing the following property of additivity for any series of disjoint sets $A_j\in \mathcal{B}_b$, $j=1,2,\ldots,$ satisfying $\bigcup_{j=1}^\infty A_j \in \mathcal{B}_b$: 
\begin{equation}
L\left(\bigcup_{j=1}^\infty A_j\right)= \sum_{j=1}^\infty L(A_j). 
\end{equation}
More specifically, a \emph{L\'evy basis} $L$ on $\mathbb{R}^d$ is a random
noise  
satisfying two conditions: 1) independent scattering: $L(A_1)\perp L(A_2)$ if $A_1\cap
  A_1=\emptyset$ for $A_1,A_2\in\mathcal{B}_b$; (2) infinite divisibility: for any $n\in\mathbb{N}$, we can
  represent $L$ as the sum of $n$ independent and identically distributed (iid) random measures
  $L_{n,1},\ldots,L_{n,n}$ such that $L\stackrel{d}{=}L_{n,1}+\ldots+L_{n,n}$. The restriction  to bounded sets $\mathcal{B}_b$ avoids handling infinite volumes.  Since bounded sets generate the full $\sigma$-algebra $\mathcal{B}$, no ambiguities or loss of generality arise.
The probability distributions $F_A$ of $L(A)$ in a L\'evy basis are infinitely divisible; i.e., for any $n\in\mathbb{N}$, a distribution $F_{n,A}$ exists such that $F_A$ is the $n$-fold convolution of $F_{A,n}$, written $F_A=F_{A,n}^{n\star}$. The L\'evy basis $L$ is stationary if $L(s+A)=L(A)$ for any
$s\in\mathbb{R}^d$ and $A\in\mathcal{B}_b$ with $s+A=\{s+y\mid y\in A\}$.   In
addition to being stationary, it is isotropic if $L(RA)=L(A)$ with
$R\in\mathbb{R}^{d\times d}$
an orthogonal rotation matrix  and  $RA=\{Rs,\ s\in A\}$.  Any setwise sum of  L\'evy bases is again
a L\'evy basis. 

For statistical applications, we focus on L\'evy bases for which $L(A)$ has tractable distribution $F_A$ for sets of interest $A$. Well-known parametric form of $F_A$ for any $A$ is given by infinitely divisible distribution families that are closed under convolution of
two iid random variables, such as the Gaussian, gamma, Cauchy, stable, Poisson or negative binomial families. The list of  infinitely
divisible parametric distributions with well-known expressions only for specific choices of $A$ is much longer, containing distributions such as the
student's $t$, Pareto, Gumbel, Fr\'echet, lognormal, inverse gamma and heavy-tailed Weibull ones. 
Other interesting constructions of infinitely divisible distributions do not possess a nice analytical form of the distribution function for any $A\not=\emptyset$; this is the case for most of the compound Poisson distributions defined through a sum of iid random variables with a Poisson-distributed number of terms. 
The characteristic function $\varphi(t;L(A))=\exp(itL(A))$  of the variables $L(A)$ in a  stationary L\'evy basis $L$ obeys the following L\'evy-Khintchine formula: 
\begin{equation}\label{eq:chf}
\varphi(t;L(A))=\exp\left(|A|\left[ita-bt^2/2+ \int_{\mathbb{R}} \left\{\exp(ity)-1-ity1_{[-1,1]}(y)\right\}  \eta(\mathrm{d}y) \right]\right), 
\end{equation}
with $|A|$ the hypervolume of $A$, and deterministic parameters  $a\in\mathbb{R}$, $b\geq 0$  and  $\eta(\cdot)$ the  so-called \emph{L\'evy measure} satisfying $\eta(\{0\})=0$ and $\int_{\mathbb{R}} \min(y^2,1))\eta(\mathrm{d}y)<\infty$.
Conversely, any infinitely divisible random variable $L'$ has a representation as in \eqref{eq:chf}, and we can use it as \emph{L\'evy seed} to define an isotropic stationary L\'evy basis by fixing $L(A_1)\stackrel{d}{=}L'$ if $|A_1|=1$; then Levy variables $L(A)$ for any set $A$ are determined by their characteristic function $\varphi(t;L')^{|A|}$. Given a stationary L\'evy basis $L$, we will use the notation $L'$ to designate a variable possessing same distribution as $L([0,1]^d)$, and $F'$  denotes the distribution of $L'$. In the following, we will assume that L\'evy bases are isotropic and stationary if not explicitly stated otherwise. 

 In \eqref{eq:chf}, the additive structure of $\log\varphi(t;L(A))$ reveals an additive decomposition of $L(A)$, and  of the L\'evy basis more generally, into three components. The value $a |A|$ is  deterministic position parameter of  $F_A$. If not explicitly stated otherwise, we assume $a\equiv 0$ in the following. The value of  $b$ characterizes the variance of an additive Gaussian component $W(A)\sim \mathcal{N}(0,b\,|A|)$ in  $L(A)\stackrel{d}{=}W(A)+\tilde{L}(A)$. The complementary additive term  $\tilde{L}(A)$ is independent of $W(A)$ and can be represented as the sum of  the atom masses in $A$ when considering a  mixture of Poisson point processes with atoms of mass $y$; we here refer to the latter component as the \emph{pure jump part} of the L\'evy basis. 

Examples of L\'evy bases are the Gaussian basis with  $\eta\equiv 0$ (and $a=0$ if it is centered), the Poisson basis with $b=0$, $\eta(\mathrm{d}x)=\lambda\delta_1(x)\,\mathrm{d}x$ with intensity parameter $\lambda>0$, or the gamma basis with $a=b=0$ and $\eta(\mathrm{d}x)=\alpha x^{-1}\exp(-\beta x)\,\mathrm{d}x$ where $\alpha>0$ is the shape and $\beta>0$ the rate. We give more details on the gamma and inverse Gaussian L\'evy bases, which will allow  constructing nonnegative spatially dependent processes using our indicator kernel approach.  The distributions in these L\'evy bases  are not stable since the convolution $F(\cdot)\star F(c\ (\cdot))$ is part of the basis only if $c=1$, but not for general values $c>0$. 
A stationary gamma L\'evy basis satisfies
$L(A)\sim \Gamma(\alpha|A|,\beta)$ with shape $\alpha|A|$  and rate $\beta>0$; it    has characteristic function
\begin{equation}
\varphi_{\Gamma}(t;L(A))=\left(1 \,-\, \frac{i\,t}{\beta}\right)^{-\alpha\, |A|}.
\end{equation}
A stationary inverse gaussian L\'evy basis satisfies $L(A)\sim \mathrm{IG}(\lambda|A|^2,\mu|A|)$ with shape $\lambda|A|^2$ ($\lambda>0$)  and mean $\mu|A|>0$ ($\mu>0$); it has  characteristic function   \citep{Chhikara.Folks.1988} 
\begin{equation}\label{eq:cfinvgauss}
\varphi_{\mathrm{IG}}(t;L(A))=\exp\left[|A|\mu{\frac{\lambda}{\mu^2}\left(1-\sqrt{1-\frac{2\mu^2\mathrm{i}t}{\lambda}}\right)}\right];
\end{equation}
modifying $A$ changes the mean to $\mu|A|$, while $\lambda/\mu^2$ is invariant.
For count data with overdispersion, the negative binomial (NB) distribution is an appropriate candidate, here parametrized as  $\mathrm{NB}(\mu,\theta)$ with the overdispersion parameter $\theta>0$ and its mean $\mu$; its variance is $\mu+\mu^2 /\theta$. The convolution of $\mathrm{NB}(\mu_1,\theta_1)$  and $\mathrm{NB}(\mu_2,\theta_2)$ with $\mu_1,\mu_2,\theta_1,\theta_2>0$ and $\mu_2/\mu_1=\theta_2/\theta_1$ is   $\mathrm{NB}(\mu_1+\mu_2,\theta_1+\theta_2)$, \emph{i.e.}, the convolution is additive in both parameters. A stationary NB L\'evy basis satisfies $L(A)\sim \mathrm{NB}(|A|\mu,|A|\theta)$ with  characteristic function 
\begin{equation}\label{eq:nbdens}
\varphi_{\mathrm{NB}}(t;L(A))=\left(\frac{1-p}{1-p\exp(it)}  \right)^{\theta\,|A|}, \qquad p=\mu/(\mu+\theta).
\end{equation}

\subsection{Spatially dependent processes through kernel convolution}
\label{sec:kernel}
Process convolutions using a L\'evy basis $L$ on $\mathbb{R}^d$ and an integrable 
kernel function $k:\mathbb{R}^d\times
\mathbb{R}^d\rightarrow[0,\infty)$ have found widespread interest in the
literature \citep[e.g.,][]{Higdon.1998,Higdon.2002} for  modeling data through an isotropic stationary  convolution process
\begin{equation}\label{eq:genconv}
 X(s)=\int_{\mathbb{R}^d} k(s,y) L(\mathrm{d}y), 
\end{equation}
which is well defined provided that mild integrability conditions are satisfied \citep[e.g.,][]{BarndorffNielsen.al.2014}. If the variance $\sigma^2>0$ of $L'$ is well-defined, the covariance function of $X(s)$  is
\begin{equation}
C(s)=\sigma^2\int_{\mathbb{R}^d} k(y-s)k(s)\,\mathrm{d}y.
\end{equation}
If   $k(\cdot)$ is not an \emph{indicator kernel} (i.e., cannot be represented as an indicator function $c 1_A(\cdot)$ with some set $A$ and constant $c>0$), the  univariate distribution of  $X(s)$ is in general not part of the L\'evy basis, except for the case of a stable distribution $F'$. Tractability problems with likelihood-based approaches may arise; see \citet{Wolpert.Ickstadt.1998,Brix.1999,Kozubowski.al.2013}, who sidestep this problem using simulation-based inference.  Moreover, closed-form expressions of the covariance function $C(\cdot)$ are available only for specific choices of the kernel $k(\cdot)$. 
By contrast, indicator kernels yield marginal distributions of $X(s)$ that are still within the L\'evy
 basis. For instance, the property of an exponential family such as the gamma or the inverse Gaussian L\'evy bases  allows using generalized additive modeling techniques for integrating covariate information into marginal distributions. 

Non-Gaussian L\'evy bases offer a wide panoply of moment properties, including
skewness and heavy tails with undefined moments. 
Flexible univariate and multivariate marginal distributions may be created by adding process convolutions
pertaining to different parametric families.  For example, consider the sum of a Gaussian
basis and a gamma basis, yielding so-called exponentially shifted Gaussian marginal distributions. The special case where the gamma convolution process is degenerated to an exponential random variable common to the whole space has been studied in \citet{Krupskii.al.2015,Krupskii.Genton.2016}.
\subsection{Indicator convolution processes}
We now consider the indicator convolution process $X(s)=L(s+A)$ defined over $\mathbb{R}^d$ using  a L\'evy basis $L$ on $\mathbb{R}^d$ and the indicator set $A\in\mathcal{B}_b$. If $\mathrm{Var}(L')=\sigma^2<\infty$, the covariance function of $X(s)$  is $C(s)=\sigma^2\, |A\cap (s+A)|$, which is the set covariance function of $A$ rescaled by $\sigma^2$. When second moments are not defined, we may still use the set covariance function for quantifying dependence in terms of  the ``size" of the additive random component $L(A\cap (s+A))$ that is common to two variables separated by the spatial lag $s$. To achieve spatial isotropy of $X(s)$, we further assume that $A$ is a hyperball.  If $d=2$, then $A=A_\rho=\{s \in\mathbb{R}^2 \mid \| s\| \leq \rho\}$ is a disc of  radius $\rho\geq 0$. We include points $s$ with $\|s\|=\rho$ into $A_\rho$ to  get  closed indicator sets. The resulting isotropic stationary covariance function (with $u=\|s\|$) is 
\begin{equation}\label{eq:covdisc}
C(u; r)= 
\begin{cases}
2\rho^2\left( \cos^{-1}(u/(2\rho))-(u/2\rho)\left[1-h^2/(2\rho^2)\right]^{1/2}  \right), & 0\leq u \leq 2\rho, \\
0, & u>2\rho.
\end{cases}
\end{equation}
One may use the approximation $C(u; \rho)\approx \pi \rho^2(1-u/(2\rho))_+$ \citep{Davison.Gholamrezaee.2012}. The covariance \eqref{eq:covdisc} has bounded support and decreases to $0$ at approximately linear speed. To gain in flexibility, we could choose a random radius $\rho$, but such mixture modeling may be inappropriate in practice when only a single independent replication of the process $X(s)$ is observed, such that the mixture distribution cannot be identified. Flexible covariance functions with bounded or unbounded support may be obtained depending on the mixture distribution, but then closed-form expressions of $C(u)$ are generally not available. Instead, we propose to go beyond the restrictive dependence structure based on disc-shaped indicator sets in $\mathbb{R}^2$  by smoothing a L\'evy basis in $\mathbb{R}^3$, using so-called \emph{hypograph kernels}.

\section{Hypograph kernels}
\label{sec:hypograph}

\subsection{Hypographs}
To generate isotropic stationary models, we define the indicator set  $A_H$ as the area enclosed between the surface $s_{d+1}=H(s_1,\ldots,s_d)$ of a radially symmetric density function $H\geq 0$ (called \emph{height function}) and the plane $s_{d+1}=0$, corresponding to the (positive) \emph{hypograph} of $H$ denoted by  $A_H=\{(s,h)\in \mathbb{R}^{d}\times [0,\infty)\mid 0\leq H(s)\leq h\}$. 
 To avoid identifiability issues between the distribution $F'$ and the volume of the hypograph,  we assume
that $\int_{\mathbb{R}^d} H(s)\,\mathrm{d}s=1$ such that $H$
is the probability density function of a $d$-variate spherical distribution.
Moreover, by abuse of notation, we will throughout use the convention $s+A_H=\{(s,0)+y\mid y\in A\}$ for $s\in\mathbb{R}^d$.

Flexible parametric models for $X(s)$ become available, and mportant differences arise in the dependence structure of $X(s)$ in contrast to the classical
kernel convolution approach in \eqref{eq:genconv} where the value of the kernel function is used to
locally rescale a L\'evy basis on
$\mathbb{R}^d$ before its aggregation over space.  Intuition for spatial modeling  suggests that what happens at a site $s$  should be stronger influenced by what is happening close to the site
than by what is happening farther away. Therefore, we will suppose that height functions are radially non-increasing such that $H(s_2)\leq H(s_1)$ if $\|s_2\|\geq \|s_1\|$, which will also simplify calculating covariance functions, see the following Section~\ref{sec:covar}. 
Illustrations of hypograph convolution processes in $\mathbb{R}^1$ are given in Figure~\ref{fig:simcont1d} for the Gaussian and gamma bases,  and in Figure~\ref{fig:simpois1d} for the Poisson basis; they are based on the Laplace density $H$ resulting in  an exponential correlation function; see Section~\ref{sec:covar}.  Figure~\ref{fig:simgamma2d} shows realizations of two spatial gamma convolution processes based on the bivariate spherical Gaussian density $H$ resulting in a correlation function defined in terms of the univariate Gaussian survival function,  see Section~\ref{sec:covar}. The L\'evy bases have been simulated using a discretization on a fine regular grid over $\mathbb{R}\times[0,H((0,0)^T)]$; for a short exposition of alternative simulation techniques, see Section~\ref{sec:simulation}. 

\begin{figure}
	\centering
		\includegraphics[width=3.5cm]{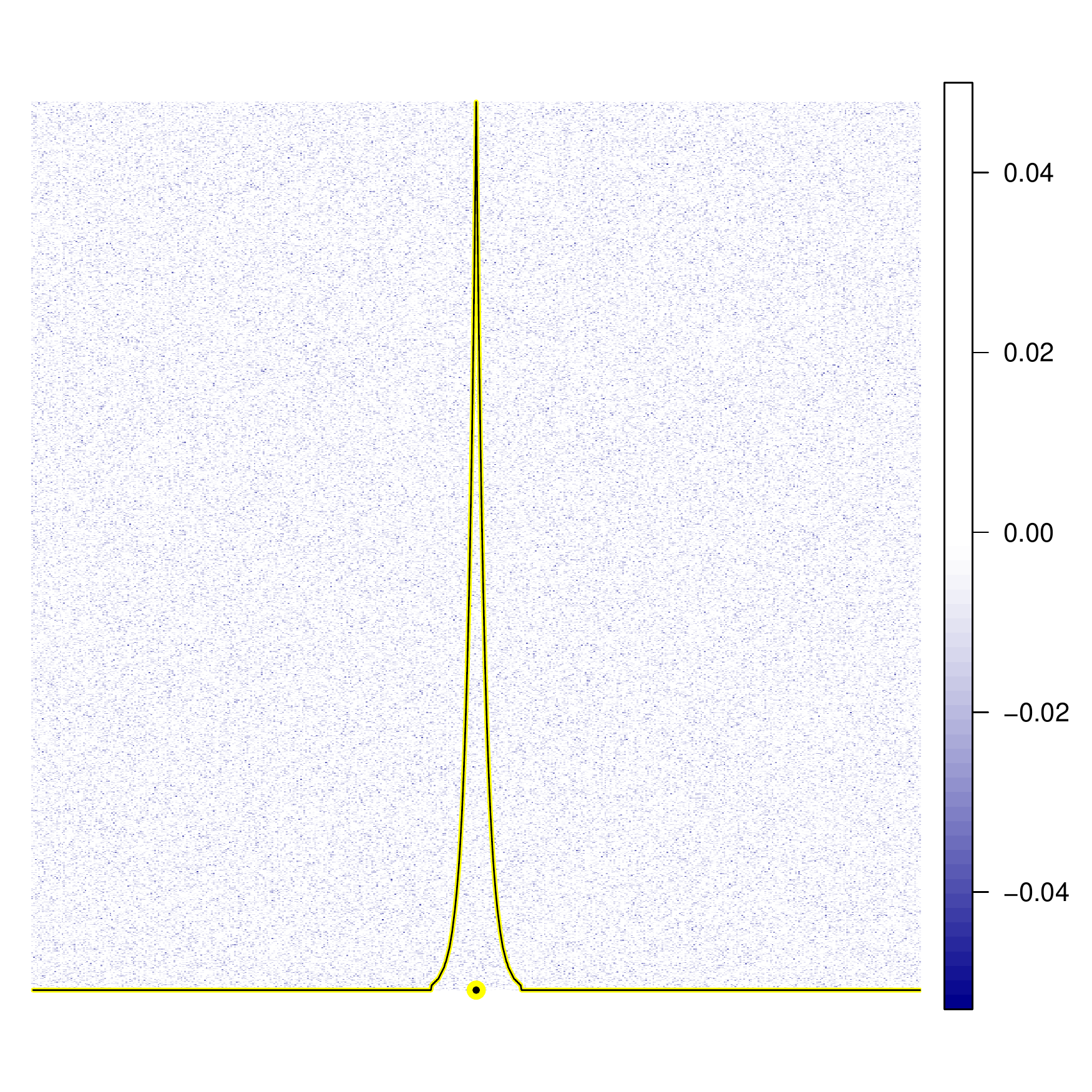}
	\includegraphics[width=3.5cm]{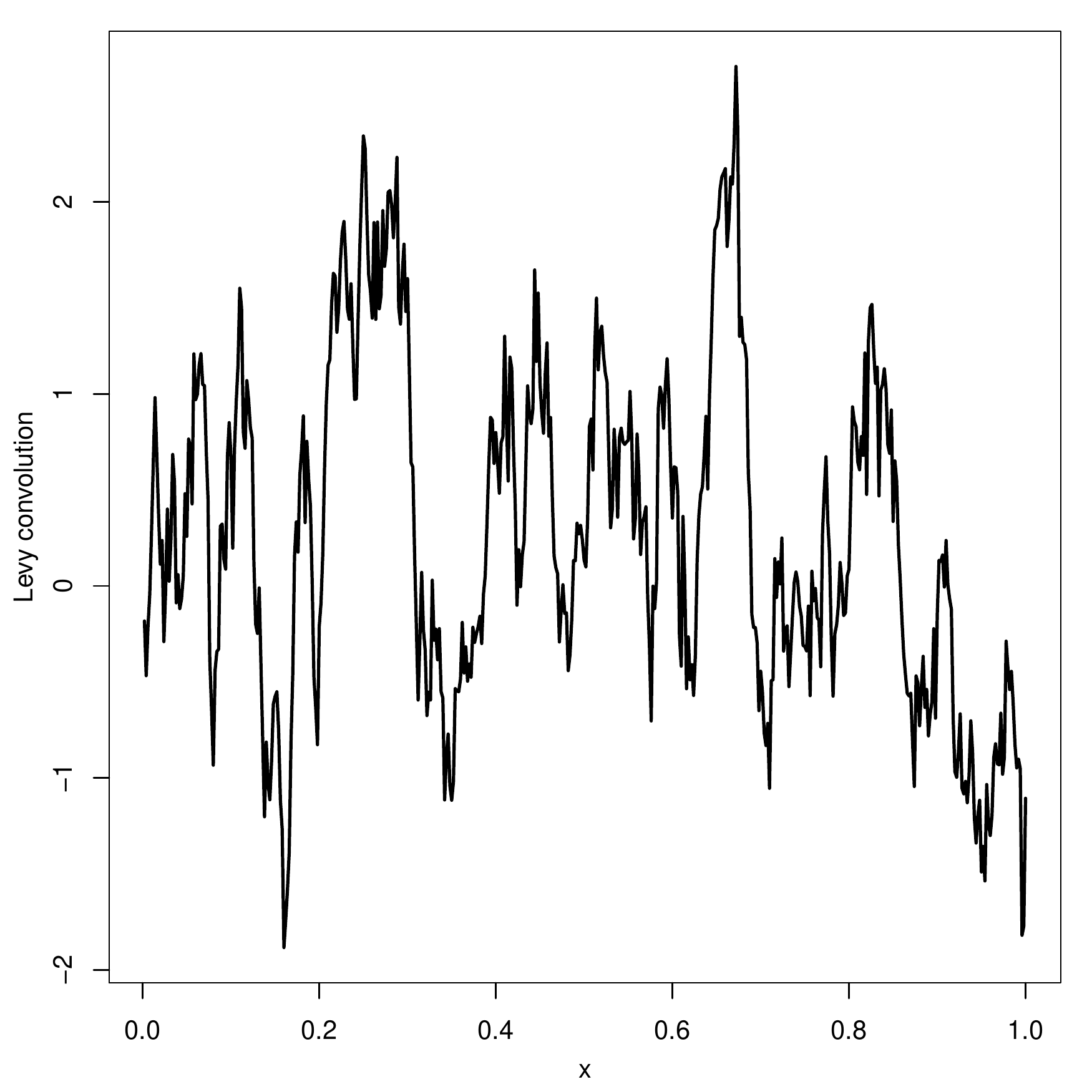}\hspace{.25cm}
	\includegraphics[width=3.5cm]{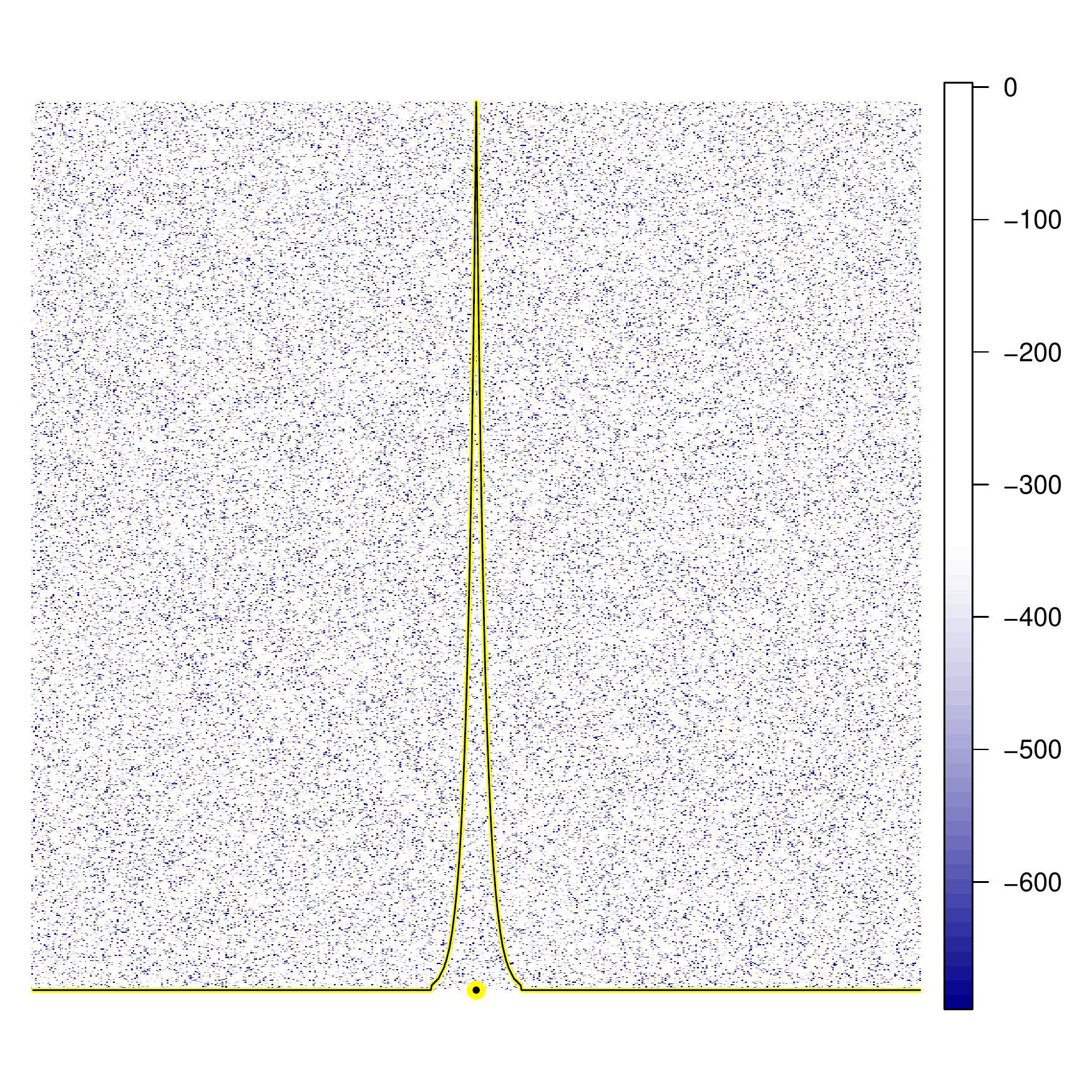}
	\includegraphics[width=3.5cm]{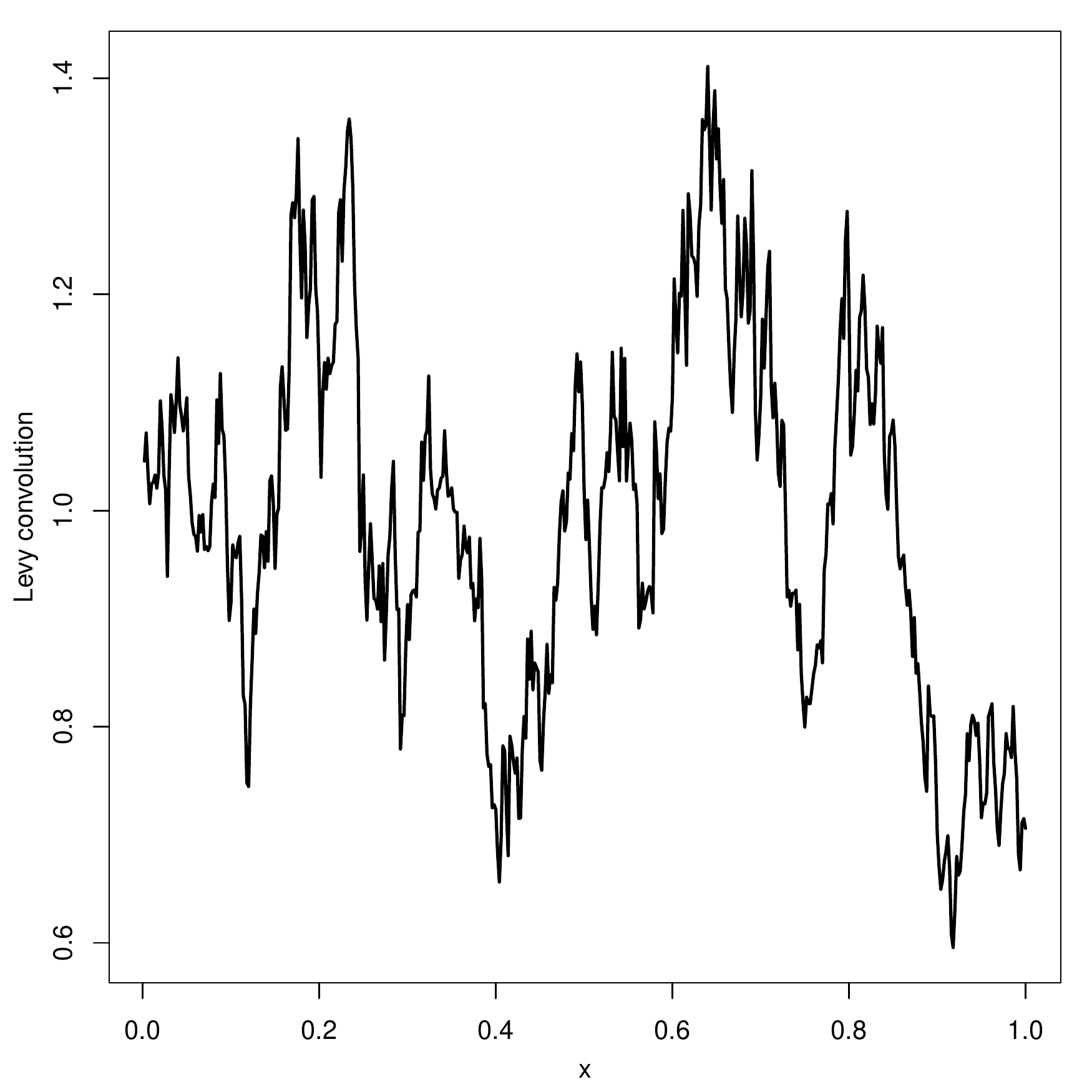}
	\caption{Hypograph convolution processes. Gaussian basis (1st display from the left) and gamma basis (3rd display, on log-scale) in $\mathbb{R}^1$ with height function given as the Laplace density.  Gaussian convolution process (2nd display), gamma convolution process (4th display).}
	\label{fig:simcont1d}
\end{figure}
\begin{figure}
	\centering
	\includegraphics[width=4.25cm]{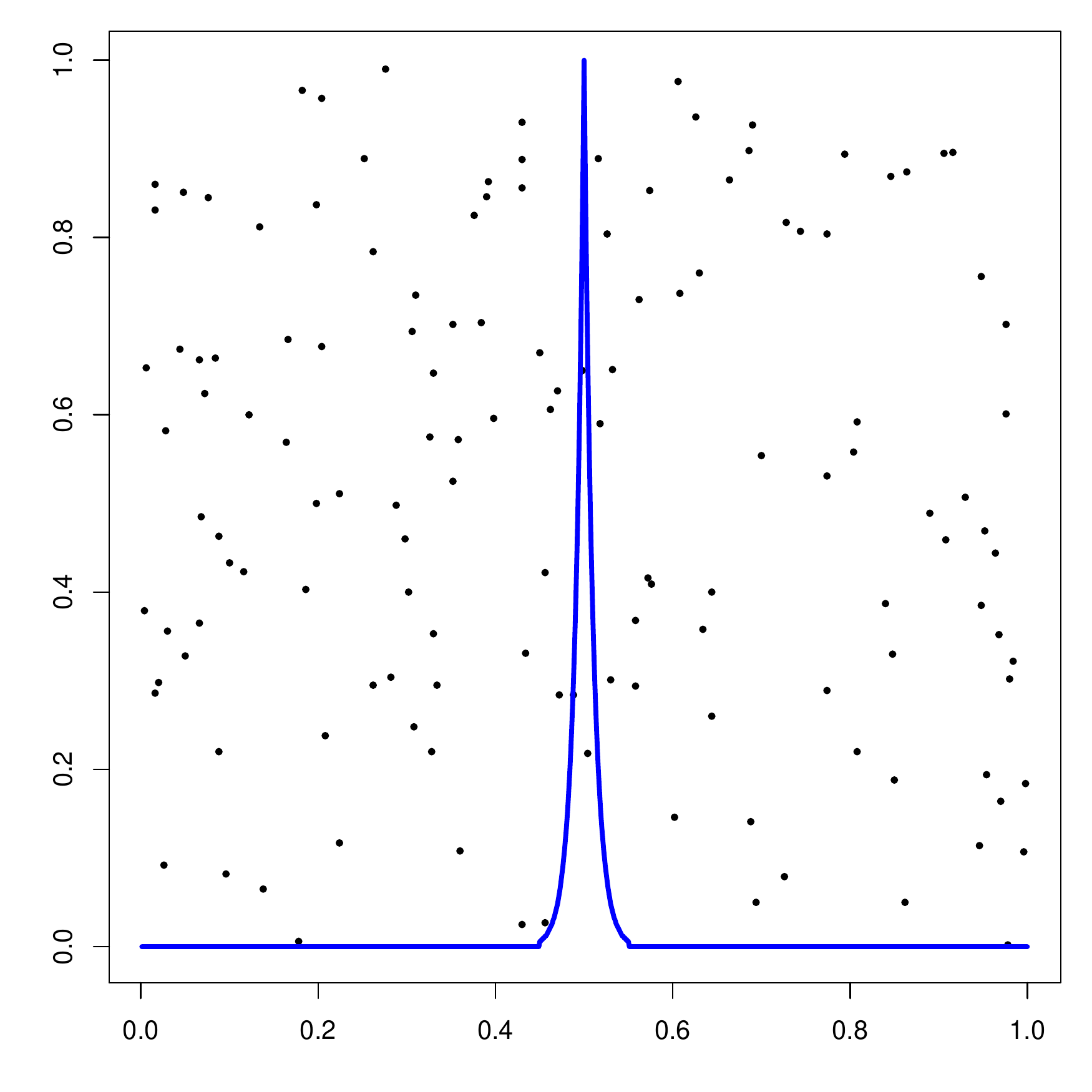}\hspace{.25cm}
	\includegraphics[width=4.25cm]{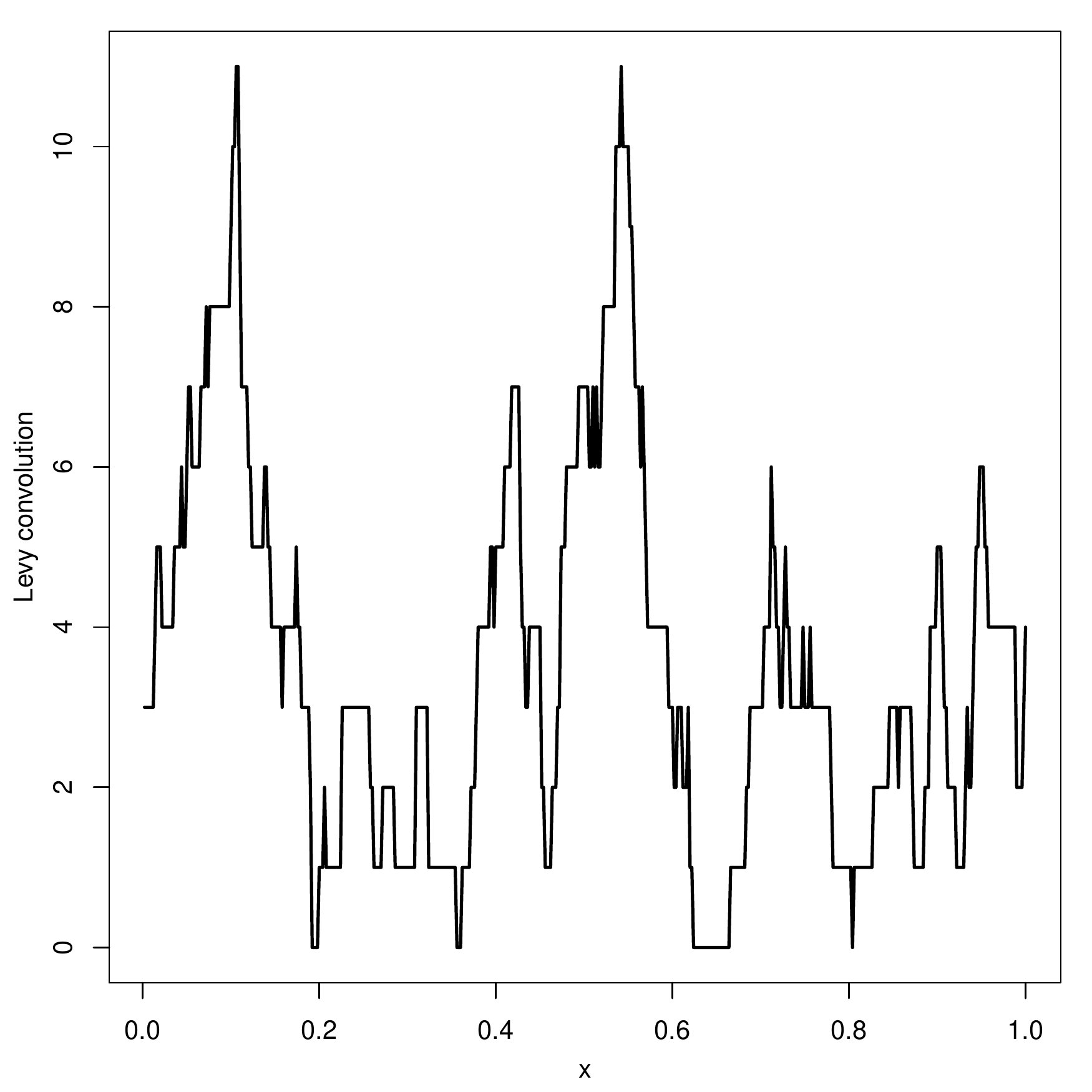}\hspace{.25cm}
	\includegraphics[width=4.25cm]{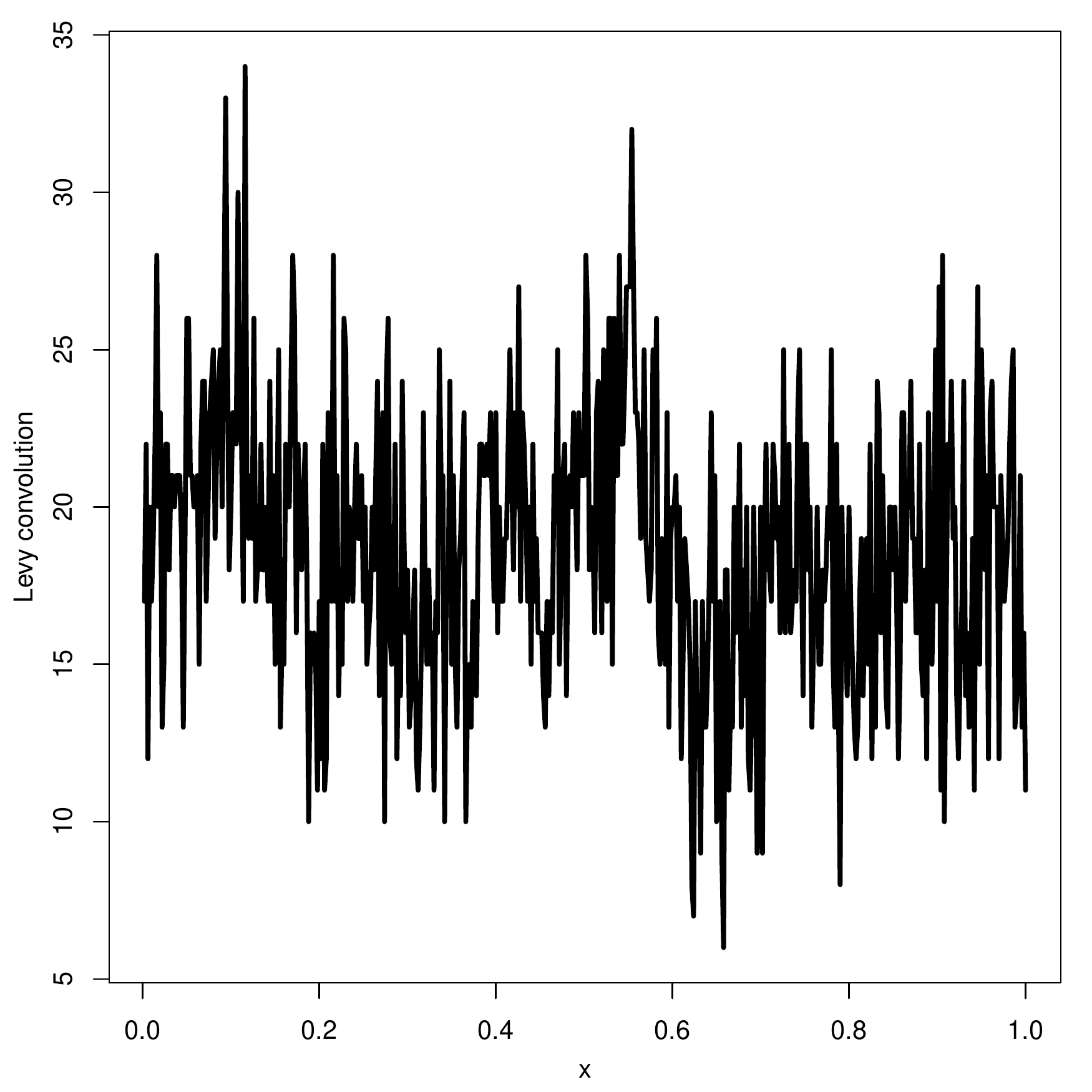}
	\caption{Hypograph convolution processes with Poisson basis (left) in $\mathbb{R}^1$ with height function given as the Laplace density. Middle: convolution process. Right: convolution process with additive Poisson nugget effect.}
	\label{fig:simpois1d}
\end{figure}

\begin{figure}
	\centering
	\includegraphics[width=6cm]{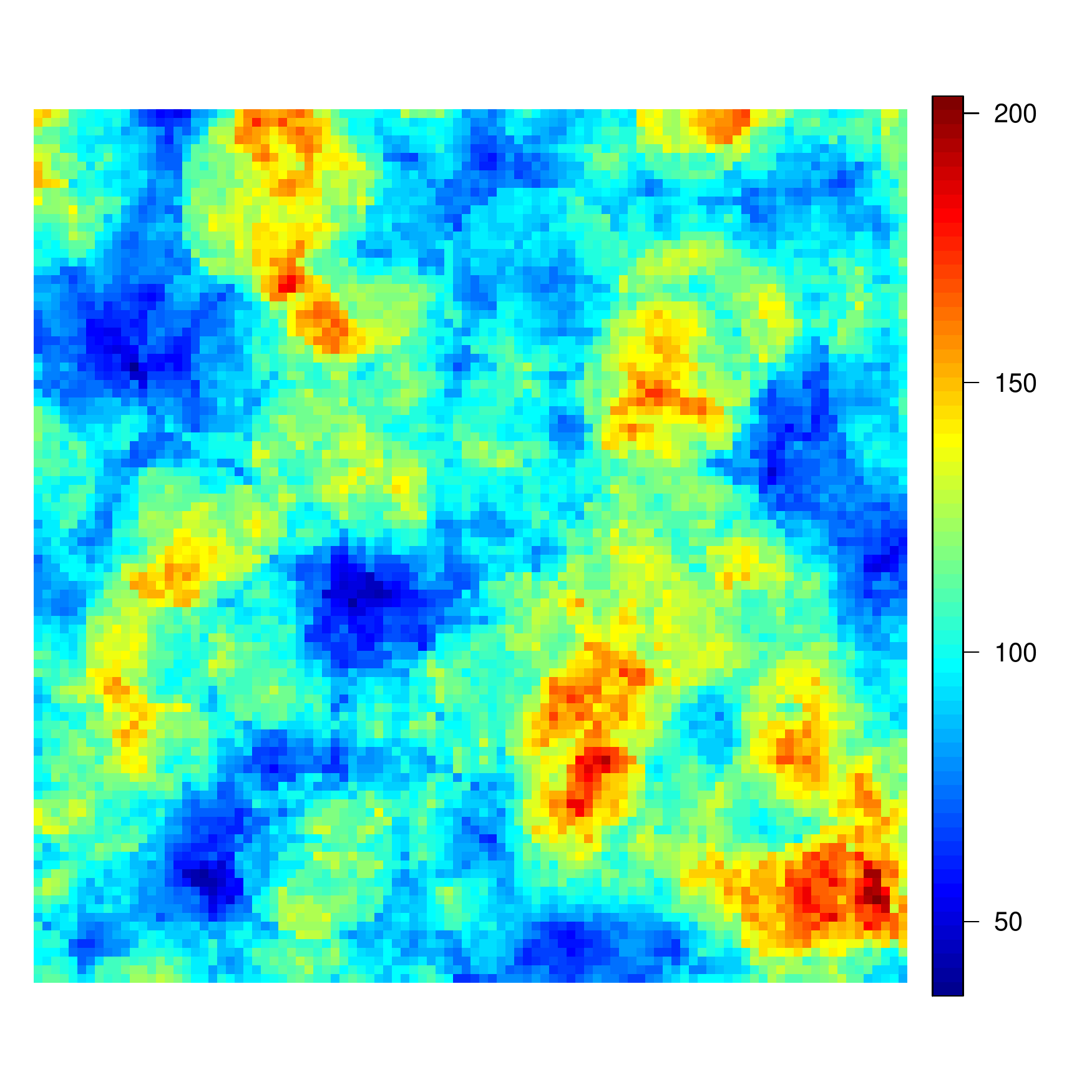}\hspace{1cm}
	\includegraphics[width=6cm]{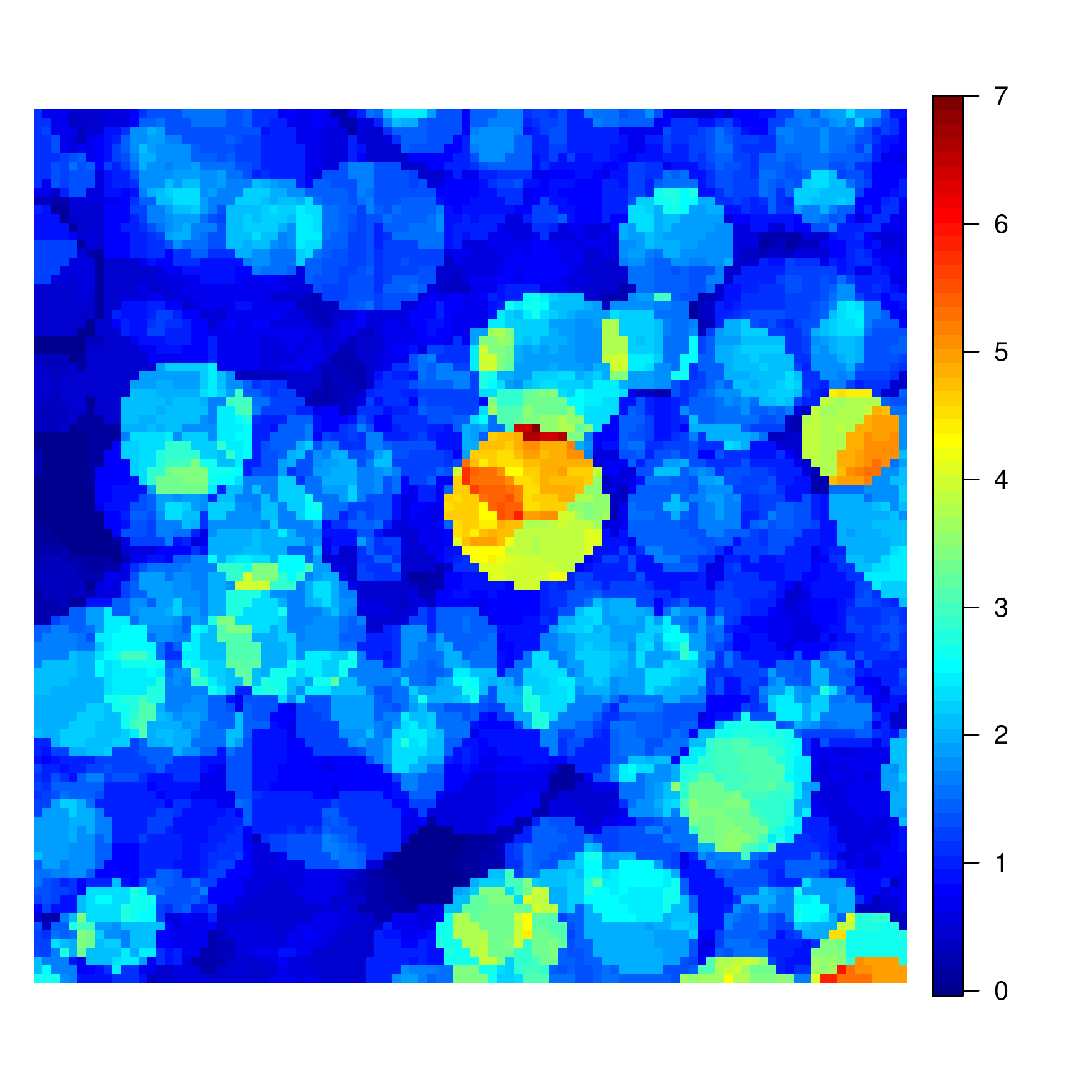}
	\caption{Hypograph convolution processes on $[0,1]^2$ with gamma L\'evy basis and height function given as the Gaussian density with standard deviation $0.05$.  Left: gamma convolution process with $\Gamma(20,1/5)$ margins. Right: gamma convolution process with $\Gamma(2,1)$ margins.}
	\label{fig:simgamma2d}
\end{figure}

\subsection{Hypograph covariance functions}
\label{sec:covar}
We now suppose that $L$ is an isotropic stationary L\'evy basis with finite second moment $\sigma^2=\mathrm{Var}(L')$, and we consider the general case of a Lebesgue-integrable height function $H_s\geq 0$ that may depend on location $s$, such that the convolution process $X(s)=L(s+A_{H_s})$ may be nonstationary. The intersection of the hypographs centered at points $s_1$ and $s_2$ is
$\{(s,h) \mid 0\leq h \leq \min(H_{s_1}(s-s_1),H_{s_2}(s-s_2))\}$. We obtain the covariance function 
\begin{equation}\label{eq:covfun}
C(s_1,s_2)=\sigma^2\int \min(H_{s_1}(y-s_1),H_{s_2}(y-s_2))\,\mathrm{d}y, \qquad s_1,s_2\in\mathbb{R}^d.
\end{equation}
Interestingly, when $H_s$ are probability densities and $\sigma^2=1$, the class of correlation functions $C$ in  \eqref{eq:covfun} 
coincides with the tail correlation functions of max-stable
processes \citep[Section 2]{Strokorb.al.2015}, which have been 
characterized in detail by \citet{Fiebig.al.2017,Strokorb.al.2015}.
In the stationary case with radially symmetric height function $H$, the covariance function is 
\begin{equation}\label{eq:covfunstat}
C(s)=\sigma^2\int \min(H(y-s),H(y))\,\mathrm{d}y, \qquad s\in\mathbb{R}^d.
\end{equation}
To obtain easily interpretable and flexible covariance
functions $C$ with closed-form expression  of the integral in \eqref{eq:covfunstat}, we  suppose that the height function $H$ is radially non-increasing. Then,  $C$ can be expressed through the univariate survival function $\overline{G}$ of the spherical
distribution characterized by $H$.  For points $s$ that are closer to $s_1$ than to $s_2$, we have
$\min(H(s-s_1),H(s-s_2))=H(s-s_2)$, and by symmetry we get
$\min(H(s-s_1),H(s-s_2))=H(s-s_1)$ when $s$ is closer to $s_2$. Writing
$u=\|s_2-s_1\|$,  the part of the intersection that is closer to $s_1$ than to $s_2$ has
hypervolume $1-G(u/2)=\overline{G}(u/2)$, and by symmetry we get the same
intersection volume for points closer to $s_2$ than to $s_1$. 
Therefore, the correlation function is
\begin{equation}
C(u)=2(1-G(u/2))=2\overline{G}(u/2),
\end{equation}
where $G(0)=0.5$  owing to radial symmetry. Next, we present  examples of interesting parametric models, ranging from  a nugget effect and bounded dependence range to the exponential correlation function and long-tailed correlation functions. 
\begin{example}[Parametric correlation functions]\label{ex:cov} We use the notation $u=\|s\|$, and we fix $d=2$ for simplicity. 
\begin{enumerate}
\item Cylinder hypograph. The cylinder-shaped height function with radius parameter $\rho>0$, given as $H(s)=1(u\leq \rho)/(\pi\rho^{2})$, yields the correlation function presented in \eqref{eq:covdisc}.
\item Half-ball hypograph. We set  $H(s)=\sqrt{\rho^2-(s_1^2+s_2^2)_+}$, such that the hypograph corresponds to a half-ball of radius $\rho>0$. Then, 
\begin{equation}
C(u;\rho)=\pi(4\rho+u)(2\rho-u)^2 / 12 \text{ for } u < 2\rho, \quad C(u)=0
\text{ for } h\geq 2\rho.
\end{equation}
\item Gaussian hypograph. Setting $H(s)=(2\pi)^{-1}\rho^{-2}\exp(-u^2/(2\rho^2))$ with $\rho>0$ yields $C(u;\rho)=2\overline{\Phi}(u/(2\rho))$ where $\Phi$ is the univariate standard Gaussian distribution function. 
\item  $t$ hypograph. Using the density of the spherical $t$ distribution with dispersion parameter $\rho>0$ and $\nu>0$ degrees of freedom yields  $C(u;\rho)=2\overline{t}_\nu(u/(2\rho))$ where $t_{\nu}$ is the distribution function of the
univariate standard $t$ distribution. 
\item Laplace hypograph.  Using the spherical Laplace distribution  with dispersion parameter $\rho>0$  \citep{Kotz.al.2001}, we get  $\overline{G}(u)=0.5\exp(-u)$ and $C(u;\rho)=\exp(-u/\rho)$. The height function $H$ has a
singularity at the origin since $H(s)\rightarrow \infty$ when
$u\rightarrow 0$. 
\item Slash hypograph. Using the spherical slash distribution with scale parameter $\rho>0$ \citep[e.g.,][]{Wang.al.2006}
yields 
 \begin{align*}
C(u;\rho)&=2(1-[\Phi(u/(2\rho))-\{\varphi(0)-\varphi(u/(2\rho))\}/(u/(2\rho))])\\&=2\Phi(-u/(2\rho))-4(\sqrt{2\pi}-\varphi(u/(2\rho)))/(u/\rho),
\quad u>0,
\end{align*}
where the limit in $0$ is $C(0)=1$. 
\item Nugget effect. The limit when $\rho\downarrow 0$ in the above examples can be interpreted as a nugget effect with the Dirac correlation function $C(u)=\delta_0(u)$.
\item Convex sums. A convex sum $\omega H_1+(1-\omega)H_2$, $\omega\in[0,1]$, of two height functions leads to the correlation function $C(u)=\omega C_1(u)+(1-\omega)C_2(u)$, with $C_1$ and $C_2$ the correlation functions associated to $H_1$ and $H_2$ respectively. 
\end{enumerate}
\end{example}


\subsection{Simulation of hypograph convolution processes}
\label{sec:simulation}
Exact simulation of the L\'evy basis $L$ and hypograph convolution processes is well understood and straightforward for the Gaussian L\'evy basis, where direct simulation of the Gaussian process $X(s)$ according to its covariance function can be done.  The Poisson L\'evy basis corresponds to a stationary spatial  Poisson process, for which exact simulation is also feasible.  For other models such as the gamma or the inverse Gaussian L\'evy bases, the additive decomposition into the Gaussian and the pure jump part based on the  characteristic function  (see Section~{sec:levybasis})  can be used for exact or approximate simulation with arbitrarily small approximation error; see \citet{Wolpert.Ickstadt.1998} for this \emph{Inverse L\'evy measure} approach. An application to simulating spatial gamma L\'evy basis has been implemented in \citet{Wolpert.Ickstadt.1998b}.

In general, an approximate realization of $(X(s_1),\ldots,X(s_m))^T$ for points $s_j$, $j=1,\ldots,m$,  can be simulated by using a discretization of the L\'evy basis on a fine grid spanning over a domain $D\times [0,h_{\max}]$ with $h_{\max}=\max_s H(s)$ (if it is finite) and $D$ including the points $s_j$. A good approximation requires that the contribution of $L((s+A_H)\cap (D\times [0,h_{\max}])^C)$ to $X(s)=L(s+A_H)$ is  negligible in practice. If densities and distribution functions are not available in closed form while the  characteristic function is, the direct numerical inversion of the characteristic function makes it possible to calculate densities and distribution functions, which then allows simulation over space discretized to a fine grid. Related tools are implemented  in \texttt{R} packages such as \texttt{CharFun} \citep{Simkova.2017} or \texttt{prob} \citep{Kerns.2017}.

Another approximation technique for smoothing a L\'evy basis in
$\mathbb{R}^{d+1}$  according to  a hypograph
$A_H$  may be to apply Cavalieri's
principle (i.e., Fubini's theorem) to reduce the problem as follows: first, calculate a finite number $m$ of independent L\'evy convolution processes $X_i(s)$ in $\mathbb{R}^d$ with disc-shaped kernels, then add them up to obtain an approximate simulation of $X(s)$. We now give some more details for $d=2$. The simulation of a $3$-dimensional L\'evy basis $L$ with L\'evy seed $L'$  and the calculation of $X(s_j)=L(s_j+A)$ are  reduced to the computationally simpler problem of    simulating $m$  L\'evy bases $L_i$ in $\mathbb{R}^2$ and smoothing them with  disc-shaped indicator kernels.
The key idea is that smoothing a L\'evy basis in $\mathbb{R}^{3}$ with a cylinder-shaped hypograph of radius $\rho$ yields the same distribution of $X(s)$ as smoothing a L\'evy basis in $\mathbb{R}^{2}$ with a disc-shaped indicator set of radius $\rho$. The L\'evy basis in $\mathbb{R}^2$ is obtained by projecting the mass of the $\mathbb{R}^{3}$ basis for $s_{3}\in[0,1/(\rho^2\pi)]$ onto $\mathbb{R}^2$. For an arbitrary continuous and radially non-increasing height function $H$, we approximate the radial function $g$ in $H((r,w))=g(r)/(2\pi)$  by a step function with $m\geq 1$ steps at radius values $r_0=0<r_1<\ldots<r_m<r_{m+1}=\infty$, with steps  of size $\rho_i=g(r_i)-g(r_{i+1})$, $i=1,\ldots,m$, and using the convention $g(r_{m+1})=0$.  We approximate $X(s)$ through the sum of $m$ hypograph convolution processes $X_i(s)$ with cylinder-shaped height functions of radius $\rho_i$, and the L\'evy seed $L'_i$ of $L_i$ has characteristic function  $\varphi(t; L'_i)=\varphi(t;L')^{1/(\rho_i^2/\pi)}$. 

\section{Extremal dependence behavior}
\label{sec:tail}
A major benefit of non-Gaussian process convolutions is increased tail flexibility.  We here show  important results with respect to bivariate extremal dependence summaries. For any two functions $a$ and $b$ with $b(x)\not=0$, we use the notation $a(x)\sim b(x)$ to indicate that $a(x)/b(x)\rightarrow 1$ when  $x\rightarrow\infty$. 
We define the \emph{tail correlation} $\chi$ of two random variables $X_1\sim F_1$ and $X_2\sim F_2$ as the conditional limit 
\begin{equation}
\chi=\lim_{x\rightarrow\infty} \mathrm{pr}\left(F_1(X_1)>1-1/x\mid F_2(X_2)>1-1/x\right) \in [0,1]
\end{equation}
if it exists, where $\chi>0$ indicates asymptotic dependence while $\chi=0$ corresponds to asymptotic independence \citep{Coles.al.1999}. In the case of asymptotic independence, joint tail decay rates are faster than marginal tail decay rates. Then,  more precise information for a wide range of bivariate distributions can be obtained through the Ledford--Tawn representation \citep{Ledford.Tawn.1996}
\begin{equation}\label{eq:ledfordtawn}
\mathrm{pr}(F_1(X_1)>1-1/x,F_2(X_2)>1-1/x) \sim \ell(x) x^{-1/\eta}, \qquad x\rightarrow \infty, 
\end{equation}
with the \emph{coefficient of tail dependence} $\eta\in(0,1]$ and a slowly varying function $\ell(\cdot)$, i.e., $\ell(tx)/\ell(t)\rightarrow 1$ when $t\rightarrow\infty$. An alternative yet equivalent parametrization is through $\overline{\chi}=2\eta-1\in(-1,1]$ \citep{Coles.al.1999}, and we have
\begin{equation}\label{eq:chibar}
\overline{\chi}=\lim_{x\rightarrow\infty} \frac{\log \mathrm{pr}\left(F_1(X_1)>1-1/x\right)}{\log \mathrm{pr}\left(F_1(X_1)>1-1/x,F_2(X_2)>1-1/x\right)}.
\end{equation}
Incidentally, the value of $\overline{\chi}$ is the linear correlation coefficient in the case of a bivariate Gaussian distribution. 
In our set-up of   isotropic stationary processes, we can define the summaries $\chi$ and $\overline{\chi}$ with respect to the distance  between two sites $s_1,s_2$. 

We first collect some notations and definitions to provide useful results on tail dependence in stationary indicator convolution processes $X(s)$ based on a L\'evy basis $L$. 
 To simplify notations, we write $X_1=L(A)$, $X_2=L(s+A)$ with $X_1\stackrel{d}{=}X_2$ the variables defined at
sites $s_1=0$ and $s_2=s$ based on their respective indicator sets
$A$ and $s+A$. Further, let $X_{12}=L(A\cap (s+A))$,
$X_{1\setminus 2}=L(A\setminus (s+A))$ and $X_{2\setminus
	1}=L((s+A)\setminus A)$. Clearly, these three variables $X_{12}$, 
$X_{1\setminus 2}$, $X_{2\setminus 1}$ are stochastically
independent with 
\begin{equation}\label{eq:bvrep}
X_1=X_{12}+X_{1\setminus 2}, \qquad
X_2=X_{12}+X_{2\setminus 1},\qquad X_{1\setminus 2}\stackrel{d}{=}X_{2\setminus 1},
\end{equation} 
such that $X_{12}$ represents a ``factor" that is common to $X_1$ and $X_2$, while the residuals $X_{1\setminus 2}$ and $X_{2\setminus 1}$ of
$X_1$ and $X_2$ respectively with respect to this factor are independent. For simplicity's sake, we use the following notation for  the hypervolumes of the indicator sets: 
\begin{equation}\label{eq:alpharep}
\alpha=|A|, \qquad \alpha_{0}=|A\cap (s+A)|, \qquad \alpha_{\mathrm{res}}=|A\setminus (s+A)|,
\end{equation}
where $\alpha_{\mathrm{res}}=|(s+A)\setminus A|$ by symmetry. We  further write $F_{\alpha}$ for the distribution of a L\'evy basis variable $L(A)$.  Moreover,  we denote by $\overline{F}(x)=1-F(x)$ the survival function of a distribution $F$, and  by $F\star F$ the distribution of its convolution with itself. We now recall important tail distribution classes, which have been established in the literature and encompass many practically relevant infinitely divisible distributions. Depending on the class to which the L\'evy basis pertains,  we will show that structually very different tail behavior arises in the convolution process. 

\emph{Subexponential distributions} are an important class of heavy-tailed distributions \citep{Foss.al.2011}.  A distribution $F$ is  called subexponential if  $\overline{F\star F}(x)/\overline{F}(x)\sim 2$, where  we additionally require that $F\star F$ is long-tailed when negative values arise with positive probability  \citep{Foss.Korshunov.2007,Watanabe.2008}. 
Infinitely divisible subexponential distributions include the  Weibull (with Weibull index smaller than $1$), lognormal, Fr\'echet or Pareto ones and, more generally, all regularly varying distributions characterized by the limit relation 
\begin{equation}
\overline{F}(tx)/\overline{F}(t)\rightarrow x^{-\gamma}, \quad t\rightarrow\infty, \quad x>0, 
\end{equation}
with the regular variation index $\gamma>0$.  If $L(A)$ is  subexponential for a set $A$, then all variables in the L\'evy basis are subexponential,  and we simply  say that the L\'evy basis is subexponential.  In this case, we get $\overline{F}_{\alpha_1}(x)/\overline{F}_{\alpha_2}(x) \sim \alpha_1/\alpha_2$ for $\alpha_1,\alpha_2>0$, see \citet{Foss.al.2011}.  

An important class of light-tailed distribution are those with exponential tails.  We say that  $F$ has \emph{exponential tail with rate $\beta>0$} if 
\begin{equation}\overline{F}(t+x)/\overline{F}(t)\rightarrow \exp(-\beta x), \qquad t\rightarrow\infty,
\end{equation}
for all $x$ \citep[e.g.,][]{Pakes.2004,Watanabe.2008}. An exponential tail in $X\sim F$ is equivalent to a regularly varying distribution of  $\exp(X)$ with index $\beta$. Two disjoint subclasses are the light-tailed convolution-equivalent distributions satisfying $\mathbb{E}\exp(\beta X)<\infty$, and the gamma-tailed distributions satisfying $\mathbb{E}\exp(\beta X)=\infty$. 
We call $F$ \emph{convolution-equivalent with rate $\beta>0$} if it is exponential-tailed with rate $\beta$ and 
\begin{equation}\label{eq:conveq}
\overline{F\star F}(x)/\overline{F}(x)\rightarrow 2 m_F(\beta)<\infty, \qquad m_F(\beta)=\int_{-\infty}^{\infty} \exp(\beta x) F(\mathrm{d}x), 
\end{equation}
where  $m_F(\beta)$ is the \emph{$\beta$-exponentiated moment}. Any positive limit arising in \eqref{eq:conveq} is necessarily $2m_F(\beta)$.  All variables $L(A)$ in a L\'evy basis with convolution-equivalent L\'evy seed $L'$ are convolution-quivalent with the same rate $\beta$ \citep[Theorem 3.1,][]{Pakes.2004}, and we say that such a L\'evy basis is convolution-equivalent.
For instance, the inverse Gaussian distribution characterized in \eqref{eq:cfinvgauss} is convolution-equivalent with $\beta=\lambda/(2\mu^2)$. 
In the exponential-tailed case where the $\beta$-exponentiated moment is not finite, we say that $F$ is \emph{gamma-tailed  with rate $\beta$}  if 
\begin{equation}\label{eq:exptype1}
\overline{F}(x) \sim \ell(x) x^{\alpha-1}\exp(-\beta x), \qquad \alpha,\beta>0,
\end{equation}
with some slowly varying function $\ell$.
For instance, the variables $L(A)$ in a gamma L\'evy basis are gamma-tailed with common rate parameter $\beta>0$ and $\ell(x)=\Gamma(\alpha'\alpha)^{-1}$ if the L\'evy seed satisfies $L'\sim \Gamma(\alpha',\beta)$, and we say that such a L\'evy basis is gamma-tailed. Distributions with tail behavior \eqref{eq:exptype1}, but where $\alpha<0$ is negative, are convolution-equivalent  \citep[Lemma 2.3,][]{Pakes.2004}.
The following Proposition~\ref{prop:se} treats the heavy-tailed set-up of subexponential L\'evy bases, for which asymptotic dependence arises in the convolution process $X(s)$.

 
  \begin{proposition}[Asymptotic dependence in subexponential L\'evy indicator convolutions]\label{prop:se}
Suppose that $L$ is a subexponential L\'evy basis.
With  notations as in \eqref{eq:bvrep} and \eqref{eq:alpharep}, the tail dependence coefficient of  the variables $X_1$ and $X_2$ for $A\not=\emptyset$ is 
\begin{equation}
\chi(X_1,X_2)= \frac{\alpha_{0}}{\alpha}.
\end{equation}
 \end{proposition}
\noindent The value of $\chi(X_1,X_2)$ can be interpreted as a variant of the set correlation function of $A$  when $A$ is shifted by the vector $(s,0)$. Interestingly, the resulting tail correlation function $\chi(s)$ does not depend on the distribution of $L$ except for the property of  subexponentiality. 
Provided that $L'$ has finite second moment, the linear correlation function and the tail correlation coincide for L\'evy indicator convolutions, which transcribes the fact that co-movements in heavy-tailed variables are strongly determined by common additive components and tail events in particular.  Next, we present a result for the gamma-tailed L\'evy bases, which possess asymptotic independence in the convolution process $X(s)$. 
\begin{proposition}[Asymptotic independence in gamma-tailed L\'evy indicator  convolutions]\label{prop:gamma}
	Suppose that $L$ is a gamma-tailed L\'evy basis where distributions $F_\alpha$  have rate parameter $\beta>0$, shape parameter $\alpha>0$ and slowly varying function $\ell_\alpha$ as defined in  \eqref{eq:exptype1}.
	With  notations as in \eqref{eq:bvrep} and \eqref{eq:alpharep}, the variables $X_1$ and $X_2$ are asymptotically independent if $s\not=0$, i.e. $\chi(s)=\delta_0(s)$, and we have $\overline{\chi}(X_1,X_2)=1$. Moreover, if $X_1$ and $X_2$ are nonnegative, then 
\begin{equation}\label{eq:expcond}
\mathrm{pr}(X_2>x\mid X_1>x) \sim \mathbb{E}\exp(\beta\min(X_{1\setminus 2},X_{2\setminus 1}))\frac{\Gamma(\alpha)}{\beta\Gamma(\alpha_0)\Gamma(\alpha_{\mathrm{res}})\ell_{\mathrm{res}}(x)} x^{-\alpha_{\mathrm{res}}}.
\end{equation}
\end{proposition}
\noindent Although gamma-tailed indicator convolution processes are asymptotically independent according to Proposition~\ref{prop:gamma}, the bivariate joint tail decay rate with $\eta(X_1,X_2)=1$ is only moderately faster than the univariate tail decay rate. 
\begin{example}[Gamma L\'evy basis]\label{ex:gamma}
	For the gamma L\'evy basis with $F_\alpha=\Gamma(\alpha,\beta)$, we have  $\overline{\chi}(X_1,X_2)=1$ and 
	\begin{equation*}
	\mathrm{pr}(X_2>x\mid X_1>x) \sim \mathbb{E}\exp(\beta\min(X_{1\setminus 2},X_{2\setminus 1}))\frac{\Gamma(\alpha)}{\beta \Gamma(\alpha_0) }x^{-\alpha_{\mathrm{res}}}.
	\end{equation*}
\end{example}
\noindent Finally, we check the light-tailed convolution-equivalent indicator convolutions, for which asymptotic dependence arises in $X(s)$.
\begin{proposition}[Asymptotic dependence in light-tailed convolution-equivalent L\'evy indicator convolutions]\label{prop:ce}
	Suppose that $L$ is a convolution-equivalent L\'evy basis with rate parameter $\beta>0$ as defined in  \eqref{eq:exptype}.
With  notations as in \eqref{eq:bvrep} and \eqref{eq:alpharep}, the variables $X_1$ and $X_2$ are asymptotically dependent if $\alpha_0>0$.  Given the $\beta$-exponentiated moments   $m_{L'}(\beta)$  of the L\'evy seed $L'$ and $\tilde{m}(\beta)$  of $\min(X_{1\setminus 2},X_{2\setminus 1})$, we have
\begin{equation}\label{eq:conveqchi}
\chi(X_1,X_2) =\frac{\alpha_0}{\alpha} \tilde{m}(\beta) m_{L'}(\beta) ^{-\alpha_{\mathrm{res}}}. 
\end{equation}
\end{proposition}
\noindent For an example, we consider the convolution-equivalent inverse Gaussian L\'evy basis. 
\begin{example}[Inverse Gaussian L\'evy basis]\label{ex:ig}
The inverse Gaussian L\'evy basis, for which  $F_{\alpha}=\mathrm{IG}(\lambda,\mu)$ with $L'\sim \mathrm{IG}(\lambda,\mu_0)$,  $\lambda,\mu_0>0$ and $\mu=\alpha \mu_0$,   is known to be convolution-equivalent with rate $\lambda/(2\mu_0^2)$.
	Using Proposition~\ref{prop:ce} and 	$m_{L'}(\lambda/(2\mu_0^2))=\exp(\lambda/\mu_0)$, we obtain 
	$$\chi(X_1,X_2)=\frac{\alpha_0}{\alpha}\tilde{m}(\lambda/(2\mu_0^2))\exp(-\alpha_{\mathrm{res}}\lambda/\mu_0),$$ yielding asymptotic dependence if $\alpha_0>0$.
\end{example}

Interestingly, gamma-tailed L\'evy bases yield asymptotic independence, while both the lighter-tailed convolution-equivalent bases and the heavier-tailed  subexponential bases yield asymptotic dependence. Hence,  a certain discontinuity arises in the tail dependence behavior when moving from lighter to heavier tails.  Moreover, exponential-tailed L\'evy bases can lead to both scenarios of asymptotic dependence or asymptotic independence in $X(s)$.

\section{Modeling extensions}
\label{sec:specificmodels}
We discuss useful extensions beyond stationary and isotropic modeling. 
For obtaining nonstationary continuous marginal distributions, location-dependent shifting or rescaling of the convolution process $X(s)$ are straightforward approaches, and a combination with generalized additive modeling to capture covariate effects is possible.
Nonstationary dependence can be achieved through location-dependent hypographs $A_{H_s}$, but  calculations of intersecting hypervolumes  may become more intricate.    A spatially independent component  can be added to the convolution process to account for a nugget effect or observation errors. 

Assuming isotropy is inappropriate when directional effects arise in spatial processes.  
We here adapt geometric anisotropy to the context of non-Gaussian L\'evy convolutions by supposing that the isotropic convolution model applies after a rotation and rescaling of the coordinate system.  Therefore, original coordinates $\tilde{s}$ can be transformed to isotropic coordinates $s$ using a rotation angle $\theta\in[0,\pi)$ and a stretching $b\geq 1$ along this direction,
\begin{equation}\label{eq:geomaniso}
s=\begin{pmatrix} b & 0 \\ 0 & 1\end{pmatrix}\begin{pmatrix}\cos(\theta)
& -\sin(\theta) \\ \sin(\theta) & \cos(\theta) \end{pmatrix} 
\tilde{s}.
\end{equation}
Alternatively, one may conduct the geometric anisotropy transformation on indicator sets by defining the height function $H$ as the density of an elliptically contoured probability distribution. However, calculating intersecting hypervolumes becomes more intricate and is related to the set covariance of ellipsoids.

The following two subsections focus on space-time modeling and on  hierarchical modeling based on embedding a latent gamma indicator convolution respectively. 


\subsection{Space-time modeling}
We propose two easily implementable possibilities to include temporal dependence into spatial models with data replicated in time, here with spatial dimension  $d=2$ for simplicity. The first model features a transport effect resulting from moving the hypograph through  space according to a velocity vector $v=(v_1,v_2)^T$,  such that $X(s,t)= L(s+vt+A_H)$ When the height function $H$ is isotropic and radially decreasing and $\sigma^2=\mathrm{Var}(L')$, the covariance function is
\begin{equation}
C(s,t) = \sigma^2 \int_{\mathbb{R}^2} \min\left(H(0),H(s-vt)\right)\,\mathrm{d}s=2\sigma^2\overline{G}(s-vt), 
\end{equation} 
 in analogy to \eqref{eq:covfun}. In general, $C$ is nonseparable between space and time. Our second space-time separable model is for observations on a regular time grid, indexed without loss of generality by $t=0,1,2,\ldots$  For fixed $t$, we require that the spatial process $X(s,t)$ is a L\'evy hypograph convolution with respect to a hypograph $H$.   To define the temporal innovation structure, we use an additive decomposition of each spatial process for fixed $t$ based on an iid series of spatial L\'evy convolution processes. The time-dependent process $X(s,t)$ is constructed by adding components from $t,t-1,\ldots$ to generate the process at time $t$.
 Fur this purpose, we consider a discrete-time kernel $k_T(i)$,  $i=0,1,\ldots$, normalized  such that $\sum_{i=0}^\infty k_T(i)=1$; i.e., $k_T$ is the probability mass function of a count distribution such as the Poisson, negative binomial or Zipf ones. Moreover, we assume monotonocity, $k_T(i+1)\leq k_T(i)$, which makes sense  when dependence strength decreases with time lag.  If $i_{\max}<\infty$ exists such that $k_T(i_{\max})>0$ and $k_T(i)=0$ for $i>i_{max}$, the  process $X(s,t)$ will be $i_{max}$-dependent over time.   We start with a sequence of iid L\'evy bases $L_t$, $t\in\mathbb{Z}$. Next, we define L\'evy hypograph convolutions  $\{\varepsilon_{t,i}(s)\}$, $t\in\mathbb{Z}$, $i=0,1,\ldots$, by applying the  height functions $k_T(i)H(\cdot)$ to each $L_t$, rescaled according to the time lag $i\geq 0$.  
The space-time convolution process is constructed as
\begin{equation}
X(s,t)=\sum_{i=0}^\infty \varepsilon_{t-i,i}(s).
\end{equation}
The spatial covariance function $C_S$ is associated to the hypograph $A_H$. Thanks to the monotonicity of  $k_T$,  we can calculate the separable space-time covariance function  for  $t_2>t_1$, 
\begin{align*}
C((s_1,t_1),(s_2,t_2))&= \mathrm{Cov}\left(\sum_{i=0}^\infty \varepsilon_{t_1,t_2-t_1+i}(s_1),\sum_{i=0}^\infty \varepsilon_{t_1,t_2-t_1+i}(s_2)\right)\\
&=\sum_{i=0}^\infty \mathrm{Cov}\left(\varepsilon_{t_1,t_2-t_1+i}(s_1),\varepsilon_{t_1,t_2-t_1+i}(s_2) \right)\\
&=C_S(s_1,s_2)\sum_{i=0}^\infty k_T(t_2-t_1+i) \\ &=C_S(s_1,s_2)C_T(t_2-t_1),
\end{align*}
whose temporal correlation function $C_T(t_2-t_1)=\sum_{i=0}^\infty k_T(t_2-t_1+i)$ is  the survival function of the distribution $k_T$.

\subsection{Latent Gamma process models}
For hierarchical modeling, we may embed a L\'evy convolution process for a parameter related to the central tendency of the univariate data distribution.  Specifically, a latent gamma indicator convolution process $G(s)$ provides useful hierarchical models.  Embedding $G(s)$ for the rate of an exponential distribution yields generalized Pareto margins. \citet{Bacro.al.2017} have used this to construct space-time models for asymptotically independent threshold exceedances by using indicator sets defined as slated cylinders in space-time, and closed-form expressions of bivariate distributions arise. In their construction, the  space-time gamma indicator convolution  $G(s,t)$ drives both the exceedances  and the exceedance probability $p(s,t)=\exp(-X(s,t))$. 
Embedding $G(s)$ for the mean of a Poisson distribution yields negative binomial (NB) margins  \citep{Hilbe.2011}. The NB parameter $\theta$ in \eqref{eq:nbdens} corresponds to the gamma shape parameter. Using bivariate results for negative binomial distributions derived in \citet[][Section~8.4.2]{Cameron.Trivedi.2013}, closed-form representations of bivariate probability mass functions of the hierarchical process can be obtained. 

A spatial inverted max-stable process \citep{Wadsworth.Tawn.2012} is obtained when embedding an indicator convolution process with positive $\alpha$-stable L\'evy basis ($0<\alpha<1$) for the rate of an exponential distribution. Indeed, multivariate distributions have logistic max-stable dependence \citep{Stephenson.2009}, and the resulting convolution process has structure very similar to the inverted Reich--Shaby model \citep{Reich.Shaby.2012}. 


\section{Pair-based inference}
\label{sec:inference}
Our modeling framework provides a natural link between   marginal distributions and  dependence structure, which avoids a full separation of margins and dependence such as in copula modeling. We therefore suppose that the data process can directly be modeled by  a L\'evy indicator convolution process, or after applying an easily tractable and interpretable marginal transformation, or by a latent L\'evy process model.  
Knowing  the parametric family of marginal distributions allows us to separate the estimation of L\'evy basis parameters from those related to the shape of the kernel, and approaches such as the independence likelihood \citep{Varin.al.2011} can be used to estimate the L\'evy basis parameters and to select an appropriate model. 
 If  the covariance function of
the L\'evy indicator convolution process $X(s)$ is well defined, standard geostatistical techniques are available to estimate parameters by contrasting empirical and model covariance functions \citep{Chiles.Delfiner.2009}. This method may also be used to provide good starting values for iteratively maximizing pairwise likelihood functions, which will be the focus of the remainder of this section. 



Pairwise likelihood approaches are composite
likelihood techniques, whose estimation efficiency and
asymptotic properties are 
close to classical maximum likelihood estimation under mild
conditions \citep{Varin.al.2011}. 
%
Bivariate vectors in L\'evy indicator convolutions can be expressed through three
independent components $X_{12}$,  $X_{1\setminus 2}$and $X_{2\setminus 1}$ defined in \eqref{eq:bvrep} and \eqref{eq:alpharep}.  If densities of these base variables are
 fast to compute, pairwise densities can computed by integrating out one of the components, here  chosen as $X_{12}$. If $\theta$ denotes the parameter vector to be estimated, the likelihood contribution of a pair
$(x_1,x_2)$ observed  at $s_1$ and $s_2$ amounts to
\begin{equation}\label{eq:pwlik}
 \ell(\theta;x_1,x_2)=\begin{cases}\int_\mathbb{R} f_{X_{1\setminus 2}}(x_1-y) f_{X_{2\setminus 1}}(x_2-y)  f_{X_{12}}(y)\,\mathrm{d}y, & \alpha_0\not= 0,\\
 f_{X_{1}}(x_1) f_{X_{2}}(x_2), & \alpha_0=0.
 \end{cases}
\end{equation}
In the case of a nonnegative integer-valued L\'evy basis for modeling count data,  the integral is replaced by a finite sum
 and  $\ell(x_1,x_2)$ can always be calculated exactly:
\begin{equation}\label{eq:pwlikdisc}
 \ell(\theta;x_1,x_2) =\sum_{y=0}^{\min(x_1,x_2)} f_{X_{1\setminus 2}}(x_1-y) f_{X_{2\setminus 1}}(x_2-y)  f_{X_{12}}(y).
\end{equation}
With continous L\'evy basis distributions, we
can always calculate the pairwise likelihood \eqref{eq:pwlik} through
numerical integration, while closed-form expression are available
in some cases. In some cases, variants of pairwise likelihood allow us to
bypass the calculation of the numerical integrals in \eqref{eq:pwlik},
which may be rather costly when data sets are large. The difference likelihood for the difference of variables $X(s_2)-X(s_1)=X_{1\setminus 2}-X_{2\setminus 1}$ is related to the difference of two  independent variables and may be better tractable.  For instance, with 
gamma L\'evy indicator convolutions it corresponds  to the difference of two iid gamma variables, a special
case of the variance-gamma distribution  \citep{Madan.Seneta.1990}.
A closed-form density can be written in terms of the modified Bessel function of the second
kind $K_\nu$. If $L'\sim\Gamma(\alpha',\beta)$, $\tilde{\alpha}=\alpha'\alpha_1$ and $x=x_2-x_1$,  we get
\begin{equation}\label{eq:pwdiff}
\ell_{\Gamma,\mathrm{diff}}(\theta;x)=\frac{\beta^{2\tilde{\alpha}}|x|^{\tilde{\alpha}-1/2}K_{\tilde{\alpha}-1/2}(\beta |x|)}{\sqrt{\pi}\Gamma(\tilde{\alpha}) (2\beta)^{\tilde{\alpha}-1/2}}, \qquad x\in\mathbb{R}.
\end{equation}
In general, estimation performance of such pairwise difference likelihoods remains comparable to the more classical
pairwise marginal likelihood,
see the study of \citet{Bevilacqua.Gaetan.2015} in the context of Gaussian processes. 

\section{Application examples}\label{sec:application}
\subsection{Bjertorp farm weed counts}
Figure~\ref{fig:weeds} illustrates the data consisting of weed counts for $100$ areal units of an agricultural field at the Bjertop farm in Sweden \citep{Guillot.al.2009}. Due to spatial dependence, we have no structure of independent replication for this data. The count sample has a mean of $81$ with an empirical standard deviation of $61$, hinting at strong overdispersion. For  isotropic stationary modeling, we consider L\'evy hypograph convolutions of  Poisson or negative binomial (NB) type, and we further allow for a nugget of Poisson or NB type respectively. Estimation is done by numerical maximization of the pairwise likelihood (PL) in \eqref{eq:pwlikdisc}  using all pairs. For the height function $H$, we use the  bivariate spherical probability densities of cylinder-shaped,  Gaussian, Laplace or Cauchy type; the latter is a student's t density with degree of freedom $\nu=1$ and has power-law tail implying long-range dependence. Our hypograph models have only  one parameter for the range to account for the small sample size with only one temporal replicate. For model selection, we  focus on maximum PL values since,  owing to strong spatial dependence, calculating  formal criteria such as the composite likelihood information criterion \citep[CLIC,][]{Varin.Vidoni.2005} would be intricate even when based on block bootstrap techniques. 

Models with Poisson L\'evy basis have much lower PL values (unreported) with relatively small variation between different models as compared to the NB basis; we attribute this to the strong overdispersion observed empirically, and we therefore do not report estimates for Poisson models. Table \ref{tab:app1} summarizes estimated hypograph parameters characterizing dependence and reports PL values for NB models. For each hypograph model, fitted with or without nugget, we have considered two estimation techniques: either with all parameters estimated in a single step, or with the negative binomial mean $\mu$ fixed to the empirical mean of observations when estimating the remaining parameters through PL. Table \ref{tab:app1} reports only the results for the single step estimation since  differences in log-likelihood values and estimated values between the two estimation procedures turned out to be small. A plausible explanation is that the mean parameter is very well identifiable even in our setting with relatively few data, which is also confirmed by small  differences in estimated means over all models.

Optimal log-PL values are very similar except for the Cauchy model, whose values  are by approximately $20$ lower than the others. The best model turns out to be the cylinder-shaped hypograph with nugget, whose relative additive contribution to the mean is estimated to be $0.21$. 
With Cauchy and Laplace models, finding good starting values for identifying a nugget value that improves upon the model without nugget was not possible, and we report an estimate of  approximately $0$.
The count proportion of $0.21$ explained by the nugget reduces the spatial dependence in comparison to the same model with nugget set to $0$; this effect is reflected in our estimates by a larger dependence range $42.2$ with nugget while it is $36.3$ without. Model fits suggest that hypographs with relativey low values at larger distances perform slightly better in comparison to the Laplace (exponential decay) and Cauchy (power decay) ones, suggesting that dependence is relatively strong at small distances but then decays strongly. An interpretation is  that the seeds at the origin of the observed weeds were dispersed groupwise (e.g., through moving air masses) with a typical within-group scattering range. Alternative, seeds  from existing weeds in this agricultural field may have been dispersed within a relatively small and well defined range, leading to spatial clusters of weeds and grouping patterns.  We refer to  \citet{Soubeyrand.al.2011} for more details on kernel-based modeling theory for group patterns, which can be characterized by stochastic models with spatial dependence. 

\begin{table}
\centering
\begin{tabular}{lrrrrr}
	{\bf Hypograph} & {\bf log-PL} & {\bf mean $\hat{\mu}$} & {\bf scale $\hat{\rho}$} & {\bf overdisp.} $\hat{\theta}$ & {\bf rel. nugget}  \\
	\hline\hline
\multirow{ 2}{*}{cylinder} &-53020&82.8&36.3&0.0185&$-$\\

&{\bf -53016}&82.7&42.2&0.0185&0.21\\
\hline
\multirow{ 2}{*}{Gaussian}  &-53019&82.7&22.5&0.0185&$-$\\

&-53018&82.7&24.3&0.0185&0.1\\
\hline
\multirow{ 2}{*}{Laplace} &-53038&82.8&6.14&0.0185&$-$\\

&-53038&82.8&6.14&0.0185&0\\
\hline
\multirow{ 2}{*}{Cauchy} &-53021&82.7&33.9&0.0185&$-$\\

&-53021&82.7&33.9&0.0185&0\\
\end{tabular}
\caption{Maximum pairwise log-likelihood and estimates for the negative binomial L\'evy hypograph convolution models fitted to the Bjertop weed count data.}
\label{tab:app1}
\end{table}

\begin{figure}
	\centering
	\includegraphics[width=0.5\linewidth]{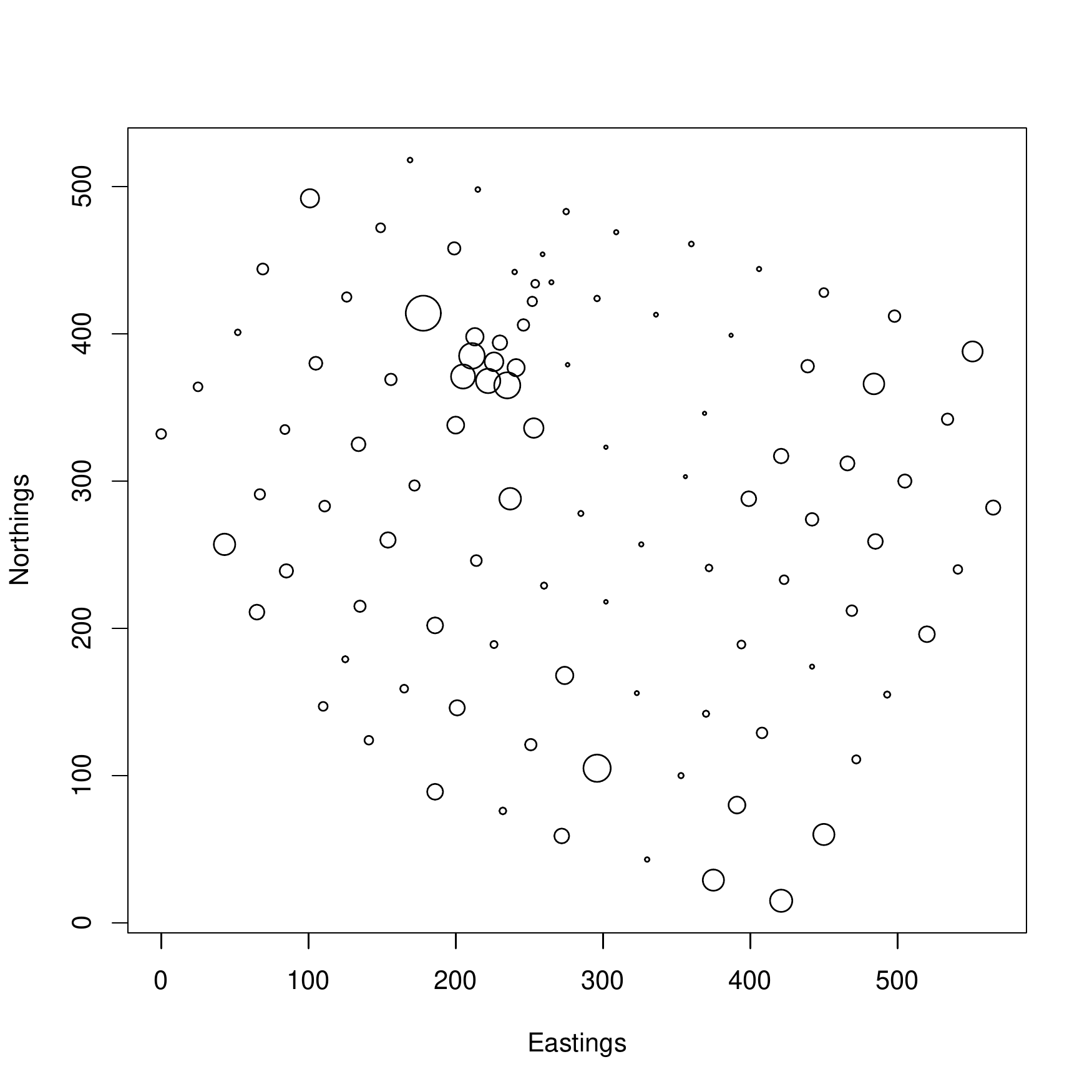}
	\caption{Spatial weed count data from the Bjertop farm (Sweden). Disc surfaces are proportional to counts.}
	\label{fig:weeds}
\end{figure}

\subsection{Daily wind speeds in the Netherlands}
We here propose spatial dependence modeling for  daily maximum wind gust data, collected between November 14, 1999 and November 13, 2008  at $30$ measurement sites spread over the Netherlands (see \url{www.knmi.nl} for public access to those data). Extreme value studies of wind speed data \citep{Ledford.Tawn.1996,Opitz.2016,Huser.al.2017} show strong support for asymptotic independence. 
For the present data, \citet{Opitz.2016} provided arguments in favor of an asymptotically independent Gaussian scale mixture model with slower joint tail decay than for Gaussian processes. In general, the Weibull distribution provides a satisfactory fit for the univariate distributions of wind speed data \citep[e.g., ][]{Stevens.Smulders.1979,Seguro.Lambert.2000,Akadaug.al.2009}, such that we here propose to power-transform a spatial gamma L\'evy convolution process with exponential margins to model the data and their spatial dependence.  As shown in Example \ref{ex:gamma}, gamma convolutions have relatively strong dependence in asymptotic dependence with coefficient $\overline{\chi}=1$, and can therefore be considered as suitable models in the light of the empirical findings of the aforementioned studies. Similar to \citet{Opitz.2016}, we use the independence likelihood to fit the marginal Weibull model with estimated Weibull index $2.63$ and the log-scale estimated as $2.69-0.0017\times \mathrm{dist}$, where $\mathrm{dist}$ is the distance to the North Sea coast. Then, data are station-wise pretransformed to a standard exponential marginal scale using the fitted Weibull distribution.

 We fix $L'\sim\Gamma(1,1)$, such that marginal distributions in the fitted hypograph convolution processes are standard exponential. We  fit hypographs of cylinder-shaped, Gaussian, Laplace, Cauchy or student's $t$ type for the dependence, and we test both isotropic or geometrically anisotropic models according to the space transformation presented in  \eqref{eq:geomaniso}. The pairwise  difference likelihood in \eqref{eq:pwdiff} is used, and estimation results are reported in Table~\ref{tab:app2}. According to the composite likelihood information criterion \citep[CLIC,][]{Varin.Vidoni.2005}, here estimated through a block bootstrap approach \citep{Carlstein.1986},  the Cauchy hypograph model comes out best, suggesting long-range dependence in wind speeds. We point out that CLIC gains due to  anisotropy tend to be more important than differences between hypographs for anisotropic models, which is in line with results of \citet{Opitz.2016} obtained for threshold exceedance data.
\begin{table}
	\centering
	\begin{tabular}{lrrrrr}
{\bf Hypograph} & {\bf scale} & {\bf nugget} & {\bf aniso:angle} $\hat{\theta}$ & {\bf aniso:scale} $\hat{b}$ & CLIC \\
\hline\hline
\multirow{ 2}{*}{cylinder}&648(42)&0.286(0.0056)&$-$($-$)&$-$($-$)&717700\\

&754(41)&0.284(0.0054)&1.24(0.06)&1.46(0.052)&717335\\
\hline
\multirow{ 2}{*}{Gaussian}&513(31)&0.286(0.0063)&$-$($-$)&$-$($-$)&717702\\

&596(38)&0.283(0.0052)&1.24(0.058)&1.46(0.049)&717435\\
\hline
\multirow{ 2}{*}{Laplace}&1160(72)&0.283(0.0047)&$-$($-$)&$-$($-$)&717539\\

&1360(89)&0.281(0.0058)&1.25(0.062)&1.47(0.058)&717493\\
\hline
\multirow{ 2}{*}{Cauchy}&402(27)&0.285(0.0056)&$-$($-$)&$-$($-$)&717707\\

&467(28)&0.283(0.0049)&1.25(0.064)&1.47(0.056)&{\bf 717320}\\
\hline
\multirow{ 2}{*}{student's $t$}&339(29)&0.284(0.0086)&$-$($-$)&$-$($-$)&717691\\

&467(28)&0.283(0.0052)&1.25(0.057)&1.46(0.053)&717345\\
%
%
%
%
%
%
%
%
%
\end{tabular}
\caption{Estimation results for gamma L\'evy hypograph convolution models fitted to the Netherlands wind speed data. Standard errors based on a block bootstrap are given in parentheses. The last columns reports the CLIC with lower values indicating better fit.  The estimate $\hat{1/\nu}$ for the inverse degree of freedom parameter of student's $t$ model is $1.83(4.58)$ for the isotropic model and $0.985(0.334)$ for the anisotropic model, implying a heavier-tailed covariance function when neglecting anisotropy in data.}
	\label{tab:app2}
\end{table}

\section{Discussion}\label{sec:discussion}
We have developed a flexible and tractable modeling framework based on L\'evy bases smoothed by indicator kernels, which allows working with distribution  properties related to tail structure and dependence  that go far beyond the ubiquitous Gaussian processes. The practical potential of such processes to bridge asymptotic dependence classes  in a natural way through the choice of the L\'evy basis family should be further studied; we refer to \citet{Wadsworth.al.2017} for some background on models covering both asymptotic dependence and independence. Moreover, sums of two convolution processes, one with bounded depedence range and asymptotic dependence, the other with asymptotic independence, would yield random field models where the asymptotic dependence range is bounded.

Spatial modeling of count data becomes feasible using pairwise likelihood, where  latent process constructions are optional but not necessary. 
With L\'evy bases suitable for count data such as the Poisson one, our approach  generalizes the multivariate models of \citet{Karlis.al.2005} and the time series models of \citet{BarndorffNielsen.al.2014} to the spatial or space-time set-up. A physical interpretation of indicator kernels is straightforward when observed count values have been aggregated over overlapping areas of the study region. 
 \citet{Wakefield.2006} states that ``there are currently no simple ways of fitting frequentist fixed-effects, nonlinear models with spatially dependent residuals". An extension of count modeling based on L\'evy indicator convolutions, used either directly or as a Poisson mean in hierarchical approaches, towards including covariates in a flexible nonstationary model appears to be a promising solution to this problem and is part of prospected work. 

More generally, using L\'evy indicator convolutions to obtain spatially dependent residuals in regression modeling paves the way towards tractable frequentist inference based on composite likelihood.
Nonstationary spatial convolution processes such as those envisaged for regression modeling can be generated by using nonstationary kernels, a nonstationary L\'evy basis or deterministic rescaling or shifting of a first-order stationary convolution process, and several of such techniques may be combined. Future work could explore the model properties and efficient inference for such models. 

Efforts should also go into fast and accurate simulation techniques, in particular for conditional simulation, which is more intricate than for the Gaussian case. This could further pave the way for simulation-based Bayesian inference of parameters, which would be an important alternative in cases where estimation uncertainty is high, such as in our application to spatial weed counts without temporal replication.


%
\appendix
\section*{Appendix}
\begin{proof}[Proof of Proposition~\ref{prop:se}]
	Using the definition of $\chi$, we get 
	\begin{equation*}
	\chi(X_1,X_2)=\lim_{x\rightarrow\infty} \frac{\overline{F_{\alpha_{ _0}} \star F_{\min(X_1,X_2)}}(x)}{F_\alpha(x)}, \qquad x\rightarrow\infty, 
	\end{equation*}
	where $\overline{F}_{\min(X_1,X_2)}(x)\sim \overline{F}^2_{\alpha_{\mathrm{res}}}(x)$.
	The tail property $\overline{F}_{\alpha_1}(x)/\overline{F}_{\alpha_2}(x)\sim \alpha_1/\alpha_2$ of subexponential distributions gives $\overline{F}_{\min(X_1,X_2)}(x)/\overline{F}_{\alpha_0}(x)\rightarrow 0$, and applying \citet[Theorem 9,][]{Foss.Korshunov.2007} then yields $\overline{F_{\alpha_{ _0}} \star F_{\min(X_1,X_2)}}(x)\sim \overline{F}_{\alpha_{ _0}}(x)$. Using this, we prove the assertion $\chi(X_1,X_2)=\alpha_0/\alpha$. 
\end{proof}

To prove the joint tail decay results in the gamma-tailed case, we first recall a result on gamma-tailed convolutions in the following Lemma~\ref{lem:exp}.
\begin{lemma}[Convolution of exponential-type random variable, see Theorem 1.1 of \citet{Hashorva.Li.2014}]\label{lem:exp}
	For two nonnegative gamma-tailed distributions $F_1,F_2$ satisfying
	\begin{equation}\label{eq:exptype}
	\overline{F}_i(x) \sim \ell_i(x) x^{\alpha_i-1}\exp(-\beta x), \qquad \beta>0,\ \alpha_i>0, \quad i=1,2,
	\end{equation}
	with slowly varying functions $\ell_i$, 
	we get 
	\begin{equation}\label{eq:exptypeconv}
	\overline{F_1\star F_2}(x)\sim \beta \frac{\Gamma(\alpha_1)\Gamma(\alpha_2)}{\Gamma(\alpha_1+\alpha_2)} \ell_1(x)\ell_2(x)x^{\alpha_1+\alpha_2-1}\exp(-\beta x), \quad x\rightarrow\infty.  
	\end{equation}
\end{lemma}
\begin{proof}[Proof of Proposition~\ref{prop:gamma}] Since $\mathrm{pr}\left(\min(X_{1\setminus 2},X_{2\setminus 1})>x\right) \sim \mathrm{pr}(X_{1\setminus 2}>x)^2$, the minimum of $X_{1\setminus 2}$ and $X_{2\setminus 1}$ is exponential-tailed with rate $2\beta>\beta$. Applying Breiman's lemma  to $\exp(\min(X_1,X_2))$, see  \citet[][Lemma 2.1]{Pakes.2004}, we obtain 
	\begin{equation}\label{eq:expmin}
	\mathrm{pr}(\min(X_1,X_2)>x) \sim \mathbb{E}\exp\left(\beta \min(X_{1\setminus 2},X_{2\setminus 1})\right)\, \mathrm{pr}(X_{12}>x).
	\end{equation}
	We calculate the limit $\overline{\chi}(X_1,X_2)$ in \eqref{eq:chibar} using \eqref{eq:expmin} and \eqref{eq:exptype}, which yields $\overline{\chi}=\eta=1$. 
	Next, Lemma \ref{lem:exp} confirms that $\alpha=\alpha_0+\alpha_{\mathrm{res}}$,  and it proves 
	\begin{align}
	\mathrm{pr}(X_1>x)=\mathrm{pr}(X_{12}+X_{1\setminus 2}>x)&\sim \beta \frac{\Gamma(\alpha_0)\Gamma(\alpha_{\mathrm{res}})}{\Gamma(\alpha)}\ell_0(x)\ell_{\mathrm{res}}(x)x^{\alpha-1}\exp(-\beta x) \notag\\
	&=\beta \frac{\Gamma(\alpha_0)\Gamma(\alpha_{\mathrm{res}})}{\Gamma(\alpha)}\ell_{\mathrm{res}}(x)x^{\alpha_{\mathrm{res}}-1}\mathrm{pr}(X_{12}>x) ,\label{eq:expmarg}
	\end{align}
	with $\ell_{0}$  the slowly varying function in \eqref{eq:exptype1} of $X_{12}$.
	By injecting formulas \eqref{eq:expmarg} and \eqref{eq:expmin} into the conditional probability
	\begin{equation}
	\mathrm{pr}(X_2>x\mid X_1>x) =\frac{\mathrm{pr}\left(\min(X_1,X_2)>x\right)}{\mathrm{pr}\left(X_1>x\right)},
	\end{equation}
	Equation \eqref{eq:expcond} follows, and clearly $\chi(X_1,X_2)=0$ if $\alpha_{\mathrm{res}}>0$, that is, if $s\not=0$.
\end{proof}
\begin{proof}[Proof of Proposition~\ref{prop:ce}]
	\citet[][Theorem 3.1]{Pakes.2004} shows that the survival function of a convolution-equivalent infinitely divisible distribution $F_{\alpha}$ is tail-equivalent to its L\'evy measure $\alpha \eta(\cdot)$ (as defined by \eqref{eq:chf})  in the following sense: 
	\begin{equation}\label{eq:conveqtail}
	\overline{F}_{\alpha}(x) \sim  m_{F_{\alpha}}(\beta)\, \alpha\eta[x,\infty),
	\end{equation}
	where $m_{F_{\alpha}}(\beta)=m_{L'}(\beta)^{\alpha}$ owing to infinite divisibility. 
	In analogy to the proof of Proposition~\ref{prop:gamma} for gamma-tailed L\'evy bases,  the minimum of $X_{1\setminus 2}$ and $X_{2\setminus 1}$ is exponential-tailed with rate $2\beta>\beta$, and applying Breiman's lemma to $\exp(\min(X_1,X_2))$ gives
	\begin{equation}\label{eq:expmin2}
	\mathrm{pr}\left(\min(X_1,X_2)>x\right) \sim   \tilde{m}(\beta)\, \mathrm{pr}(X_{12}>x).
	\end{equation}
	By writing the tail representation \eqref{eq:conveqtail} of $X_{12}$ with $\alpha$ replaced by $\alpha_{\mathrm{res}}$ and combining \eqref{eq:conveqtail}  and \eqref{eq:expmin2}, the value of $\chi(X_1,X_2)$ can be calculated, 
	\begin{equation*}
	\chi(X_1,X_2)=\lim_{x\rightarrow\infty}\frac{\tilde{m}(\beta) \alpha_0\, m_{F_{\alpha_0}}(\beta)\, \eta[x,\infty)}{\alpha\, m_{F_{\alpha}}(\beta)\, \eta[x,\infty)}=\frac{\alpha_0}{\alpha} \tilde{m}(\beta) m_{L'}(\beta) ^{ -\alpha_{\mathrm{res}}},
	\end{equation*}
	where we have used $m_{F_{\alpha_0}}(\beta)=m_{L'}(\beta)^{\alpha_0}$.
\end{proof}



\bibliographystyle{CUP}
\bibliography{biblio}

\end{document}